\documentclass[11pt]{article}
\usepackage{graphicx}    
\usepackage{amssymb,latexsym}
\usepackage{amsmath,amsfonts,amsthm}
\usepackage{comment}
\theoremstyle{plain}			

\newtheorem{thm}{Theorem}[section]
\newtheorem{lemma}[thm]{Lemma}
\newtheorem{prop}[thm]{Proposition}
\newtheorem{cor}[thm]{Corollary}

\theoremstyle{definition}		

\newtheorem{defn}[thm]{Definition}
\newtheorem{examp}[thm]{Example}

\def\row#1#2{{#1}_1,\ldots ,{#1}_{#2}}

\begin{document}

\title{{\Large \bf{Nonconvergent Electoral Equilibria under Scoring Rules: Beyond Plurality
}}}

\author{Dodge Cahan, John McCabe-Dansted and Arkadii Slinko}
\date{}
\maketitle
\thispagestyle{empty}

\begin{abstract}
We use Hotelling's spatial model of competition to investigate the position-taking behaviour of political candidates under a class of electoral systems known as scoring rules. In a scoring rule election, voters rank all the candidates running for office, following which the candidates are assigned points according to a vector of nonincreasing scores. Convergent Nash equilibria in which all candidates adopt the same policy were characterised by Cox \cite{cox1}. Here, we investigate nonconvergent equilibria, where candidates adopt divergent policies. We identify a number of classes of scoring rules exhibiting a range of different equilibrium properties. For some of these, nonconvergent equilibria do not exist. For others, nonconvergent equilibria in which candidates cluster at positions spread across the issue space are observed. In particular, we prove that the class of convex rules does not have Nash equilibria (convergent or nonconvergent) with the exception of some derivatives of Borda rule.  Finally, we examine the special cases of four-, five- and six- candidate elections. In the former two cases, we provide a complete characterisation of nonconvergent equilibria.
\end{abstract}

\section{Introduction}
\label{intro}

Hotelling's \cite{hotelling} spatial model of competition, first introduced in 1929, has had a large and varied influence on a  number of fields. It has been applied not only in the original context of firms selecting geographic locations along ``Main Street" so as to maximise their share of the market, but also to that of producers deciding on how much variety to incorporate into their products \cite{chamberlin}. Downs \cite{downs} has  adapted it with minor modifications to model an election: in particular, the ideological position-taking behaviour of political candidates in their effort to win votes. 

The model in its simplest form features a number of candidates (firms) adopting positions on a one-dimensional manifold, usually taken to be the interval $[0,1]$, along which the voters'  ideal positions (consumers) are distributed. The preferences of the voters over the candidates are then determined by the distance between their own ideal positions and those advocated by the candidates. The candidates adopt positions so as to maximise their share of the vote (market).\footnote{See Stigler's \cite{stigler} argumentation for this assumption and a discussion of it in \cite{denzaukatsslutsky}.} In the economic interpretation of this model where firms compete for market share, competition on price is excluded and customers buy from the closest firm to minimise transportation costs. 

In the majority of the literature, the voters only have one vote, which they allocate to their favourite candidate. That is, the electoral system is taken to be plurality rule. The voters' second, third and other preferences do not come into play. Economically, this is akin to saying customers only patronise the nearest firm, with the distribution of the more distant firms being irrelevant to them. 

In many situations, however, preferences other than first do indeed matter. Far from all elections are held under plurality rule: electoral systems, both in use and theoretical, are diverse. So too are the incentives that candidates are faced with under different electoral systems---what is the optimal strategy under one system may be subpar under a different one.  One very general class of electoral systems that take into account more than just the voters' single most preferred candidate are the scoring rules. 

In an election held under a scoring rule, each voter submits a preference ranking of all $m$ candidates running for office, whereby $s_i$ points are assigned to the candidate in the $i$th position on the voter's ranking. Thus, we specify a scoring rule by an $m$-vector of real nonnegative numbers, $s=(s_1,\ldots,s_m)$, with  $s_1\geq \cdots \geq s_m\geq 0$ and $s_1>s_m$. Well-known examples of scoring rules include plurality, Borda's rule and antiplurality, given by score vectors  $s=(1,0,\ldots,0)$, $s=(m-1,m-2,\ldots,0)$ and $s=(1,\ldots,1,0)$, respectively.  Different score vectors may define the same scoring rule. In particular, an affine transformation of the score vector $s'=(\alpha s_1+\beta, \ldots, \alpha s_m+\beta)$, where $\alpha>0$ and $\beta$ are real numbers, defines the same rule as $s$.\par

In the election---according to Stigler's thesis---the candidates adopt positions with the aim of maximising the total number of points received from all the electorate and hence, one would expect, the equilibrium strategies will depend on the particular scoring rule in use. This dependence is the focus of our investigation.

This is not the only possible interpretation of the scoring rule. Another plausible view, due to Cox \cite{cox2}, is that the coordinates of the score vector represent probabilities:  if we normalise the score vector $s$ so that the sum of its coordinates is 1, $s_i$ can now be seen as the likelihood that a voter votes for the $i$th nearest candidate. Indeed, ideological proximity is not the only factor at the ballot box: a single issue out of many may put a voter off a candidate who, on the whole, is of a similar ideological bent; or, it could simply be down to personal charisma, experience, prejudice or any number of other similar nonpolicy reasons. In the economic interpretation, it is also natural to assume that consumers patronise more distant firms with some probability---occasionally, it turns out to be more convenient to purchase from a more distant firm simply because one happens to be passing through the vicinity, for example. The probability implied by such an interpretation is of an ordinal nature---there is no dependency on absolute distance, as occurs in most of the extant literature on probabilistic voting (see, e.g., \cite{coughlin,duggan}).

As mentioned, most of the literature considers the case of plurality elections. Exceptions are Cox \cite{cox1} and Myerson~\cite{myerson1}, who consider a number of alternative electoral systems, including scoring rules. Cox, in particular, provides a valuable characterisation of convergent Nash equilibria (CNE), that is, equilibria in which all candidates adopt the same policy position (see Theorem \ref{CNE} of this paper for the precise formulation). He identifies three classes to which a scoring rule may belong: ``best-rewarding" rules which never possess CNE; ``worst-punishing" rules which, to the contrary, do allow CNE in which all candidates adopt a position located within a certain interval; and, an intermediate case between the two, for which a unique CNE exists at the median voter's ideal position. 

As one infers from the nomenclature, under each of these three classes of scoring rules candidates are faced with different incentives. Myerson~\cite[p.~677]{myerson1}---who introduced the terminology we are using\footnote{ In \cite{cox2}, Cox refers to ``first-place rewarding", ``intermediate" and ``last-place punishing" rules.}---explains the competing incentives as follows:
under a best-rewarding rule, ``the candidate gains more from moving up in the preferences of voters who currently rank her near the top"; on the contrary, under a worst-punishing rule ``the candidate gains more from moving up in the preferences of voters who currently rank her near the bottom."

Cox stopped short of investigating the existence of nonconvergent Nash equilibria (NCNE) for general scoring rules, i.e., those equilibria in which the positions adopted by the candidates are not all the same. For plurality, however, the situation is well-studied and a full characterisation of NCNE is given by Eaton and Lipsey \cite{eatonlipsey} and Denzau et al. \cite{denzaukatsslutsky}. The nonexistence of NCNE for antiplurality is also known, as Cox \cite[p.~93]{cox1} points out. For general scoring rules, however, the picture is unclear, and little work has been done in this direction. Cox conjectures \cite[p.~93]{cox1} that ``nonconvergent equilibria are at best rare for [worst-punishing] scoring functions, but that they are fairly common for [best-rewarding] scoring functions." We look to fill this gap in the literature.

Since a general characterisation of NCNE appears to be intractable, our approach is to consider the problem in large classes of rules satisfying various additional conditions. For several broad classes of scoring rules, we find that NCNE do not exist at all or seldom exist. For rules having scores that are  convex (decrease more rapidly at the top end of the score vector than at the bottom end), NCNE are impossible except for some derivatives of Borda rule (Theorem \ref{convexscores2}).  For rules whose scores are concave (decrease more rapidly at the bottom end of the score vector than at the top end) we do not have such a perfect answer. Our results for this class of rules are applicable to a larger class of weakly concave rules: we prove that if an NCNE exists, then it should be highly asymmetric (Theorem~\ref{WPnoNCNE2}) and we showed that such asymmetric equilibria for some weakly concave rules do indeed exist (Example~\ref{8-4NCNE}). The existence of NCNE for rules with concave scores remain an intriguing open question. 
%
We also look at a class of rules that are highly best-rewarding, that is, the reward for being ranked first by voters is particularly large relative to other rankings. For these rules as well, no NCNE exist (Theorem \ref{highlyBR}).

On the other hand, two other classes investigated, amongst others, do allow NCNE and we can calculate many of them. We identify a class of rules allowing NCNE with multiple candidates clustered more or less evenly along the issue space (Theorem \ref{multipositional}).  We also characterize NCNE in which there are two symmetrically located clusters with the same number of candidates (Theorem \ref{bipositional}).

In the special cases of four- and five-candidate elections, things are simpler and we are able to provide a complete characterisation of the rules allowing NCNE (Theorems \ref{4cand} and \ref{5cand2}).  We also examine six-candidate elections (Theorem \ref{6cand2}). 

The rest of this paper is organised as follows. First, we introduce notation and the main assumptions of our model in Section~\ref{themodel}. Then, in Section~\ref{preliminaries}, we derive some preliminary results before moving on to investigate rules that do not allow NCNE in Section \ref{ruleswithnoNCNE}. Next, in Section~\ref{ruleswithNCNE}, we look at rules that do admit NCNE. In Section~\ref{specialcases}, we turn our attention to the special cases of four-, five- and six-candidate elections. In Section~\ref{comres} we briefly outline the algorithm we used for finding NCNE computationally. Finally, we briefly review the related literature in Section~\ref{relatedliterature}, and Section \ref{conclusion} concludes.

\section{The model}
\label{themodel}

In our model there is a continuum of voters, assumed to have ideal positions uniformly distributed\footnote{The voters' ideal positions need not actually be uniformly distributed. As has been pointed out in other papers such as \cite{aragonesxefteris}, it is enough that the candidates believe the distribution to be uniform or are working under this simplifying assumption.} on the interval $[0,1]$, the issue space, on which candidates adopt positions. Since the two- and three-candidate cases are well-known---in the first we have the classical median voter result, and in the latter case no NCNE exist (see the end of this section)---we assume there are $m\geq 4$ candidates. Candidate $i$'s position is denoted $x_i$ and a strategy profile $x=(x_1,\ldots,x_m)\in [0,1]^m$ specifies the positions adopted by all the candidates. For a given strategy profile $x$, denote by $x^1,\ldots,x^q$ the \textit{distinct} positions that appear in $x$, labelled so that $x^1<\cdots <x^q$. In an NCNE, we will always have $q\geq 2$. Let $n_i\in \mathbb{N}$ denote the number of candidates adopting position $x^i$. Then we have $\sum_{i=1}^qn_i=m$. 
Given a strategy profile $x$, we  will often use the equivalent notation
$$
x=((x^1,n_1),\ldots,(x^q,n_q)).
$$
We also use the notation $[n]=\{1,\ldots,n\}$ and if $I=[a,b]\subseteq [0,1]$ is an interval with end points $a<b$, its length is denoted $\ell(I)=b-a$ and this is the measure of voters in the interval, since the distribution is uniform. 

We assume the voters are not strategic and have single-peaked, symmetric utility functions. Hence, a voter ranks the candidates according to the distance between her own personal ideal position and the positions adopted by the candidates. If  $n_i\geq 2$, that is, two or more candidates adopt position $x^i$, then all voters will be indifferent between these candidates. However, the scoring rule requires that they submit a strict ranking. In this case, a voter separates these candidates in her ranking by means of a fair lottery. Thus, for example, if $x^i$ is the position nearest to a given voter's ideal position, the voter allocates the first $n_i$ places in her ranking randomly, with each candidate at $x^i$ having a probability of $1/n_i$ of being assigned any one of these places.\footnote{Alternatively we could allow indifferences on voters' ballots and modify the scoring rule accordingly.} Note that we assume voters can only ever be indifferent between candidates if they occupy the same position. 

Candidate $i$'s score, $v_i(x)$, is the total number of points received on integrating across all voters. The candidates are assumed to maximise $v_i(x)$---that is, they are score (share) maximisers. Information is assumed to be complete and no candidate has an advantage---they all adopt positions simultaneously.

We hope that the following example will aid the reader in understanding how scores are calculated for a given configuration of candidate positions. 

\begin{examp}\label{calculationexample} Consider a three-candidate election contested between candidates $i$, $j$ and $k$ and held under a scoring rule $s=(s_1,s_2,s_3)$. Consider the profile $x=((x^1,2),(x^2,1))$, where $x_i=x_j=x^1$ and $x_k=x^2$. This situation is illustrated in Figure \ref{examplepicture}. The voters in the interval $I_1=[0,(x^1+x^2)/2]$ are indifferent between candidates $i$ and $j$, but prefer them both to $k$. The voters in the interval $I_2=[(x^1+x^2)/2,1]$ have $k$ as their unique favourite candidate, and they are indifferent between $i$ and $j$. Then $i$ and $j$ both receive the score
$$v_i(x)=v_j(x)=\left(\frac{s_1+s_2}{2}\right)\ell(I_1)+\left(\frac{s_2+s_3}{2}\right)\ell(I_2).$$
The first term comes from the fact that all voters in $I_1$ flip a coin to decide whether to rank $i$ first and $j$ second, or the other way around. Hence, $i$ receives $s_1$ from half these voters and $s_2$ from the other half. The second term follows from similar considerations with respect to the voters in $I_2$.

\begin{figure*}
\centerline{
\mbox{\includegraphics[width=5.00in,height=1.50in]{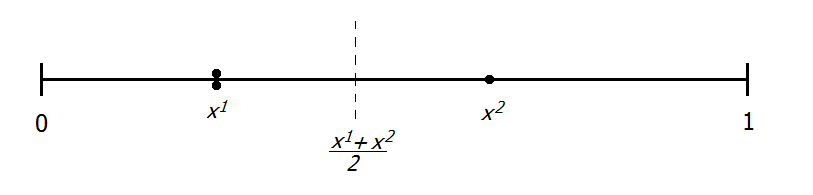}}
}
\caption{ The situation in Example \ref{calculationexample}}
\label{examplepicture}
\end{figure*}

The score garnered by $k$ is
$$v_k(x)=s_1\ell(I_2)+s_3\ell(I_1).$$
No ``sharing" of votes occurs here, since no voters are indifferent between $k$ and another candidate.

Note that $x$ is not an NCNE since profitable deviations are possible. Candidate $k$, for example, can move inwards towards the other two candidates to gain a larger score. 
\end{examp}

Often we will consider what happens to a candidate's score on making incremental deviations from her current position, or deviations to points infinitesimally close to other occupied positions. For example, given an initial profile $x=(x_1,\ldots,x_i,\ldots,x_m)$, if we write that $x'$ is the profile where candidate $i$ moves to $x_i+\epsilon$, what we mean is that $x'=(x_1,\ldots,x_i+\epsilon,\ldots,x_m)$, where $\epsilon>0$ is taken to be small enough that there are no occupied positions between $x_i$ and $x_i+\epsilon$; moreover, when $\epsilon$ is small enough, for all $j,k\in [q]$ we have $|x_i-x_j|>|x_i-x_k|$ if and only if $|x_i+\epsilon-x_j|>|x_i+\epsilon-x_k|$. This means that after this move only candidates previously at the same location as $x_i$ change their positions in the voters' rankings. As $\epsilon$ tends to zero, the measure of voters between $x_i$ and $x_i+\epsilon$ tends to zero.

Given a scoring rule $s=(s_1,\ldots,s_m)$, sometimes we will consider \textit{subrules} of $s$. A subrule of $s$ is a vector $s'=(s_i,s_{i+1},\ldots,s_{i+j})$ where $i,j\geq 1$ and $i+j\leq m$. Thus, if $s_i>s_{i+j}$, a subrule is itself a scoring rule corresponding to an election with $j+1$ candidates. If $s_i=\cdots =s_{i+j}$, then $s'$ does not define a scoring rule and we say $s'$ is a \textit{constant subrule} of $s$. An important parameter of the rule will be the average score denoted $\bar{s}=\frac{1}{m}\sum_{i=1}^ms_i$. \par\medskip


Our equilibrium concept is the standard Nash equilibrium in pure strategies. To define it we introduce the following standard notation. Let $x^*=(x_1^*,\ldots,x_{m}^*)$ be a strategy profile. Then by $(x_i,x_{-i}^*)= (x_1^*,\ldots,x_{i-1}^*,x_i,x_{i+1}^*,\ldots,x_m^*)$ we mean the profile where the $i$th candidate adopts position $x_i$ instead of $x_i^*$, whereas all other candidates do not change their strategies.

\begin{defn}\label{Nasheq} 
A strategy profile $x^*=(x_1^*,\ldots,x_m^*)$ is in Nash equilibrium if for all $i\in [m]$ we have
$$
v_i(x^*)\geq v_i(x_i,x_{-i}^*), 
$$
for all $x_i \in [0,1]$. 
A Nash equilibrium (NE) $x$ is said to be a \textit{convergent Nash equilibrium} (CNE) if all candidates adopt the same position, i.e.,  $x=((x^1,m))$. 
If, in a Nash equilibrium, at least two candidates adopt distinct positions, we say it is a \textit{nonconvergent Nash equilibrium} (NCNE).
\end{defn}

\begin{defn}
If a strategy profile 
$$
x=((x^1,n_1),\ldots,(x^q,n_q))
$$
is an NE, then we will say that $(\row nq)$ is its {\em type}.  For example, the type of a CNE is $(m)$.
\end{defn}

Before we begin presenting our results, we restate Cox \cite{cox1} characterisation of CNE for arbitrary scoring rules.

\begin{thm}[Cox \cite{cox1}]\label{CNE} Given a scoring rule $s$, the profile $x=((x^1,m))$ is a CNE if and only if 
$$c(s,m)\leq x^1\leq 1-c(s,m),$$
where $\displaystyle c(s,m)=\frac{s_1-\bar{s}}{s_1-s_m}.$
 \end{thm}
 
For the inequality in Theorem~\ref{CNE} to be satisfied for some $x^1$, it must be that $c(s,m)\leq 1/2$. The number $c(s,m)$ encodes important information about the competing incentives characteristic of a given scoring rule: in particular, it measures the first-to-average drop in the value of the points relative to the first-to-last drop. Motivated by the above theorem, Cox defined a \textit{best-rewarding rule}\footnote{Terminology of R. Myerson \cite{myerson1}.} to be a rule with $c(s,m)> 1/2$. Here, the incentive for candidates to receive first place in a voter's ranking outweighs the incentive to receive only an average one. On the other hand, if $c(s,m)< 1/2$ the rule is said to be \textit{worst-punishing}. For such rules, receiving a first-place ranking is not all that better than an average ranking---the main thing is to avoid being ranked last. Rules satisfying  $c(s,m)=1/2$ are called \textit{intermediate}. Hence, Cox's result says that a rule has CNE if and only if it is worst-punishing or intermediate, and the possible equilibrium positions are in some interval of sufficiently centrist points. 

The parameter $c(s,m)$ attains its maximum value of $c(s,m)=1-1/m$ for plurality rule, given by $s=(1,0,\ldots,0)$. Its minimum value of $c(s,m)=1/m$, on the other hand, occurs for the  antiplurality (veto) rule, given by $s=(1,\ldots,1,0)$. It is closely related to Myerson's \cite{myerson1} ``Cox threshold", which in our notation is $\bar{s}$.\par\medskip

In addition, Cox also observed that in any NCNE the most extreme positions, $x^1$ and $x^q$, must be occupied by at least two candidates and, hence, in the case of a three-candidate election, no NCNE exist.

\section{Preliminaries}
\label{preliminaries}

Suppose in a profile $x$ there are $m$ candidates $1,2,\ldots, m$ in positions $x_1,\ldots, x_m$ occupying locations $x^1,\ldots,x^q$ and suppose candidate $i$ is not alone at the location she occupies. We will now investigate how the score of candidate $i$  changes when $i$ changes  position from $x_i$ to $t\in [0,1]$ while the positions of all candidates except $i$ are fixed. That is we are interested in the function $v_i(t,x_{-i})$.

\begin{prop}
\label{insidetheinterval1}
In the intervals $(0,x^1)$ and $(x^q,1)$  the function $v_i(t,x_{-i})$ is  linear  with slopes $(s_1-s_m)/2$ and $-(s_1-s_m)/2$, respectively.
\end{prop}

\begin{proof}
Let $t\in [0,x^1)$. Let us split the issue space into intervals
\[
I_1=\left[0,\frac{t+x^1}{2}\right), \ldots, I_k=\left[\frac{t+x^{k-1}}{2},\frac{t+x^k}{2}\right),\ldots, I_{q+1}= \left[\frac{t+x^q}{2},1\right]. 
\]
The voters in $I_k$ rank candidate $i$ in the same position on their ballots. For $1\ne k\ne q+1$ the length of $I_k$ does not depend on $t$ (and hence the contribution to $i$'s score from voters in $I_k$). So when $t$ changes, only the contributions to $i$'s score from the voters in the two end intervals change. The sum of these is
\[
s_1\left(\frac{t+x^1}{2}\right) + s_m\left(1-\frac{t+x^q}{2}\right),
\]
so the result follows. The second statement follows from the first due to an obvious symmetry.
\end{proof}

\begin{prop}
\label{insidetheinterval2}
Suppose that apart from $i$ there are $j$ candidates in positions $x^1,\ldots, x^{l-1}$ and $k$ candidates in positions $x^l,\ldots, x^{q}$, where $j+k=m-1$. Then
in the interval $(x^{l-1},x^l)$   the function $v_i(t,x_{-i})$ is  linear  with slope $(s_{j+1}-s_{k+1})/2$. 
\end{prop}

\begin{proof}
Suppose candidate $i$ is currently at $t$. Let us split the issue space into intervals $[0,t]$ and $[t,1]$ (so that candidate $i$ belongs to both of them) and apply the previous proposition to rules $(\row s{j+1})$ and $(\row s{k+1})$. Then the nonconstant contributions to the score $v_i(t,x_{-i})$ from both intervals will be 
\[
\left[s_1\left(\frac{x^l-t}{2}\right) + s_{k+1}\left(1-\frac{t+x^q}{2}\right)\right] + \left[s_1\left(\frac{t-x^{l-1}}{2}\right) + s_{j+1}\left(\frac{t+x^1}{2}\right)\right]\
\]
\[
=\frac{(s_{j+1}-s_{k+1})}{2}t+\text{const}.,
\]
which proves the proposition.
\end{proof}

We also note that at points $x^1,\ldots, x^q$ the function $v_i(t,x_{-i})$ is in general discontinuous and all three values
\[
v_i(x^{j-},x_{-i}) =\lim_{t\to x^{j-}}v(t,x_{-i}),\quad v_i(x^j,x_{-i}) ,\quad v_i(x^{j+},x_{-i}) =\lim_{t\to x^{j+}}v(t,x_{-i}),
\]
where the first and the third are one-sided limits, may be different.\par\medskip

We will now derive some general properties of NCNE which will come in useful later. The following lemma places a lower bound on the number of candidates that may occupy an extreme left or an extreme right position. This is a generalisation of Cox's \cite[p.~93]{cox1} observation that there can be no less than two candidates at extreme positions.

\begin{lemma}\label{generalprops1} 
Given a scoring rule $s$, let $1\leq k\leq m-1$ be such that $s_1=\cdots =s_k>s_{k+1}$. Then a necessary condition for a profile $x$ to be in NCNE is $\min(n_1,n_q)>k$.
\end{lemma}

\begin{proof}
Suppose $n_1\leq k$ and candidate 1 is located at $x^1$. Let us split the issue space into intervals
\[
I_1=\left[0,\frac{x^1+x^2}{2}\right), \ldots, I_j=\left[\frac{x^1+x^{j}}{2},\frac{x^1+x^{j+1}}{2}\right),\ldots, I_{q}= \left[\frac{x^1+x^q}{2},1\right]. 
\]
 At $x$ the contribution to candidate 1's score $v_1(x)$ from the interval $I_1=[0,(x^1+x^2)/2]$ is 
$$
\left(\frac{1}{n_1}\sum_{i=1}^{n_1}s_i\right)\ell(I_1)=s_1\ell(I_1),
$$ 
and the contribution from the interval $I_j$ is 
$$
\left(\frac{1}{n_1}\!\!\sum_{i=k_j}^{k_j+n_1-1}\!\!\!\!\! s_i\right)\ell(J_j),
$$
for some $k_j$, since candidate 1 is tied in the rankings of all voters in $I_j$.
If 1 moves infinitesimally to the right, then these contributions to $v_1(x^{1+},x_{-1})$ become 
$s_1\ell(I_1)$ and $s_{k_j}\ell(I_j)$, respectively. Indeed, $1$ is still ranked at worst $k$th by voters in $I_1$ and hence loses nothing. Also, $1$ rises in all other voters' rankings. If $s_{k_j}>s_{k_j+n_1-1}$ for at least one $j$,  this move is strictly beneficial. This infinitesimal ``move'' is not a real move. However, if it is strictly beneficial, then by Proposition~\ref{insidetheinterval2} we may conclude that a sufficiently small move to the right will also be beneficial. 
If  $s_{k_j}=s_{k_j+n_1-1}$ for all $j$, then we consider the move by candidate 1 to the right, to a position $t\in (x^1,x^2)$. The intervals where voters rank candidate 1 similarly will then be
\[
I_1'=\left[0,\frac{t+x^2}{2}\right), \ldots, I_j'=\left[\frac{t+x^{j}}{2},\frac{t+x^{j+1}}{2}\right),\ldots, I_{q}'= \left[\frac{t+x^q}{2},1\right]. 
\]
We see that candidate 1 has increased the length of the first interval, from which she receives $s_1$, at the expense of the far-right interval, from which she receives $s_{m-n_1+1}=s_m<s_1$, while keeping the lengths of all other intervals unchanged. So this move is beneficial.

So for NCNE we must have $n_1> k$. Similarly, $n_q> k$.
\end{proof} 

This already allows us to make the following observation which rules out NCNE for a class of scoring rules.

\begin{cor}\label{generalprops2} 
If $s$ is a scoring rule such that $s_1=\cdots = s_k>s_{k+1}$ for some $k\geq \lfloor m/2 \rfloor$, then $s$ allows no NCNE.
\end{cor}

\begin{proof} By Lemma \ref{generalprops1}, $\min(n_1,n_q)>k\geq \lfloor m/2\rfloor$. Hence, $n_1+n_q>m$, a contradiction.  \end{proof}


The rules specified in Corollary~\ref{generalprops2} are usually worst-punishing. However, there are some that are slightly best-rewarding, as in the following example. This shows that there exist best-rewarding rules which, unlike plurality, do not allow NCNE. By Theorem \ref{CNE}, they have no CNE either, thus they have no Nash equilibria whatsoever.

\begin{examp} If $m$ is odd, consider $k$-approval with $k=(m-1)/2$. That is, $s=(1,\ldots,1,0,\ldots,0)$, where the first $k$ positions are ones. Then 
$$c(s,m)=1-\frac{1}{m}\left(\frac{m-1}{2}\right)=\frac{1}{2}+\frac{1}{2m}>\frac{1}{2}.$$
So the rule is best-rewarding but by Corollary \ref{generalprops2} it has no NCNE. 
Note that starting with this rule we can produce a family of best-rewarding rules that have no NCNE. Clearly, the scores $s_{k+1},\ldots,s_m$ could have been, instead of zeros, arbitrary scores with the property that $$\frac{1}{m}\sum_{i=k+1}^m\!\! s_i<\frac{1}{2m}$$ and the rule would still have been best-rewarding, but without NCNE.\end{examp}

The following lemma says that in an NCNE no candidates may occupy the most extreme positions on the issue space, namely 0 or 1. This will allow us to always assume that it is possible for a candidate to make a move into the end intervals $[0,x^1)$ and $(x^q,1]$.

\begin{lemma}\label{noextremepositions} Let $s$ be a scoring rule. In an NCNE $x$ no candidate may adopt the most extreme positions. That is, $0<x^1$ and $x^q<1$.\end{lemma}

\begin{proof} Suppose $x^1=0$. Let $1\leq k\leq m-1$ be such that $s_1=\cdots =s_k>s_{k+1}$. Then by Lemma \ref{generalprops1} we have $n_1>k$.  Candidate 1's score is
$$v_1(x)= \left(\frac{1}{n_1}\sum_{i=1}^{n_1}s_i\right)\frac{x^2}{2}+ S,$$
where $S$ is the contribution to $v_1(x)$  from voters in the interval $I=[x^2/2, 1]$. Now suppose candidate 1 moves infinitesimally to the right. Then, in the limit, 1's score is
$$ 
v_1(x^{1+},x_{-1})=s_1\frac{x^2}{2}+S', 
$$
where $S'$ is the new contribution to 1's score from the interval $I$. Since 1 has moved up in the ranking of all the voters in $I$, $S'\geq S$.
Also, since $s_1=s_k> s_{n_1}$ we have 
$$ 
v_1(x^{1+},x_{-1})=s_1\frac{x^2}{2}+S' > \left(\frac{1}{n_1}\sum_{i=1}^{n_1}s_i\right)\frac{x^2}{2}+ S'\geq \left(\frac{1}{n_1}\sum_{i=1}^{n_1}s_i\right)\frac{x^2}{2}+ S=v_1(x).
$$
By Proposition~\ref{insidetheinterval2}, $v_1(x^{1+},x_{-1})>v_1(x)$. Hence, candidate~1 benefits from moving to the right, and hence $x$ is not in NCNE. Thus, $x^1>0$. Similarly, $x^q<1$.\end{proof}

The next lemma places an upper bound on the length of the interval between two occupied positions, or between the boundary of the issue space and the nearest occupied position. From this, we will be able to derive a lower bound on the number of occupied positions for a given scoring rule.

\begin{lemma}\label{upperboundintervals} Given a scoring rule $s$, if $x$ is in NCNE, then both of the following conditions must be satisfied:
\begin{enumerate}
\item[(a)] $x^1\leq 1-c(s,m)$ and $1-x^q\leq 1-c(s,m)$;
\item[(b)] $x^i-x^{i-1} \leq 2(1-c(s,m))$ for any $i$ such that $2\leq i\leq q$.
\end{enumerate}
\end{lemma}

\begin{proof} 
Since $s_m<\bar{s}<s_1$, there exists a real number $\alpha$   with the property that $0<\alpha <1$ and
$\alpha s_1+(1-\alpha)s_m=\bar{s}$. Rearranging this equation, one verifies that $\alpha=1-c(s,m)$.

At any profile $x$, there will be at least one candidate $j$ who garners a total score $v_j(x)\leq \bar{s}$. Hence, if $I$ is an end interval, namely $[0,x^1]$ or $[x^q,1]$, then we must have $\ell(I)\leq \alpha$. To show this we assume without loss of generality that $\ell(I)>\alpha$ for $I=[0,x^1]$ and we show that candidate
$j$ would be able to make a profitable move to $x^1-\epsilon$  for some very small $\epsilon$. By Proposition~\ref{insidetheinterval1},  this would be proved if we could show that  
\[
v_j(x^{1-},x_{-j})>v_j(x).
\]
The idea is that by locating incrementally to the left from $x^1$ candidate $j$ captures $s_1$ from all the voters in $I$ and at the very worst $s_m$ from all other voters. So
$$
v_j(x^{1-},x_{-j})\geq s_1\ell(I)+s_m(1-\ell(I))>\alpha s_1+(1-\alpha)s_m=\bar{s}\geq v_j(x).
$$

Similarly, suppose $x^i-x^{i-1}> 2\alpha$ for some $2\leq i \leq q$. Note that we can assume candidate $j$ is not an unpaired candidate located at $x^i$ or $x^{i-1}$, since then $j$ already receives a score greater than $\alpha s_1+(1-\alpha)s_m$, which contradicts that $j$'s score is no more than $\bar{s}$. Hence, we may assume that position $x^i$ and $x^{i-1}$ remains occupied when $j$ moves. Now, if $j$ moves to any point $t$ in the interval $I=(x^{i-1},x^i)$, we have
$$
v_j(t,x_{-j})\geq s_1\left(\frac{x^i-x^{i-1}}{2}\right)+s_m\left(1-\frac{x^i-x^{i-1}}{2}\right)>\alpha s_1+(1-\alpha)s_m\geq v_j(x).
$$
Hence, for $x$ to be in NCNE it must be that $x^i-x^{i-1}\leq 2\alpha$. 
\end{proof}

Condition (a) of the previous lemma generalises Cox's \cite[p.~88]{cox1} argument that in a plurality election the most extreme candidates on either side are located outside the interval $(1/m,1-1/m)$. 

\begin{cor}\label{upperboundintervalscorollary} In an NCNE under a scoring rule $s$ with $c(s,m)>1/2$, the most extreme candidates on either side are located outside the interval $(1-c(s,m),c(s,m))$. \end{cor}

\begin{proof} This is simply condition (a) of Lemma \ref{upperboundintervals}. We need $c(s,m)>1/2$ for the result to be meaningful.\end{proof}

Note that the interval $(1-c(s,m),c(s,m))$ reaches its maximal width under plurality rule.

\begin{cor}\label{lowerboundq} 
Given a scoring rule $s$, in a Nash equilibrium the number of occupied positions $q$ satisfies $$q\geq \Big\lceil \frac{1}{2(1-c(s,m))}\Big\rceil.$$
\end{cor}

\begin{proof} If $c(s,m)\leq 1/2$ it is clearly true as the right-hand side is equal to 1. If $c(s,m)>1/2$, only NCNE can exist with $q\geq 2$ occupied positions. Then the issue space can be partitioned into two end intervals, $[0,x^1)$ and $[x^q,1]$, together with $q-1$ intervals of the form $[x^{i-1},x^i)$ for $2\leq i\leq q$. So, using Lemma \ref{upperboundintervals}, we see that
$$ 1=x^1+(1-x^q)+\sum_{i=2}^q(x^i-x^{i-1})\leq 2q(1-c(s,m)),$$
whence the result. We round up since $q$ is an integer.\end{proof}

Thus, the amount of dispersion observed is increasing in $c(s,m)$. When $c(s,m)$ increases above $3/4$, the number of occupied positions required for NCNE increases to at least three. When $c(s,m)$ exceeds $5/6$, the number of occupied positions must be at least four. The maximum value of $c(s,m)=1-1/m$ is attained for plurality, so there must be at least $m/2$ occupied positions. 
We will return to consider these bounds in Section \ref{highlyBRsection} where, with the help of a few more lemmas, we will find that for most rules that are sufficiently best-rewarding, there are no NCNE at all.

\begin{lemma}
\label{45cand2} 
Suppose at profile $x$ candidate $i$ is at $x^l$ and $n_l=2$. Then $v_i(x^{l-},x_{-i})+v_i(x^{l+},x_{-i})=2v_i(x)$. In particular, when $x$ is in NCNE, $v_i(x^{l-},x_{-i})=v_i(x^{l+},x_{-i})=v_i(x)$.
\end{lemma}

\begin{proof} Again, the issue space can be divided into subintervals of voters who all rank $i$ in the same position.
The immediate interval around $x^l$ is $I_l=I_l^L\cup I_l^R$, 
where: $I_l^L=[(x^{l-1}+x^l)/2,x^l]$ if $l>1$ or $I_l^L=[0,x^1]$ if $l=1$; and, $I_l^R=[x^l,(x^{l}+x^{l+1})/2]$ if $l<q$ or $I_l^R=[x^l,1]$ if $l=q$. 
The contribution to $v_i(x)$ from the interval $I_l$ is
$$
 \frac{s_1+s_2}{2}(\ell(I_l^L)+\ell(I_l^R)).
$$
The contribution to $v_i(x^{l-},x_{-i})$ from this interval is then
$ s_1\ell(I_l^L)+s_2\ell(I_l^R)$ and the contribution to $v_i(x^{l+},x_{-i})$ is
$ s_2\ell(I_l^L)+s_1\ell(I_l^R).$

The contribution to $v_i(x)$ from any interval $J$  to the left of $I_l$, consisting of voters who all rank $i$ similarly, is
$$ \frac{s_t+s_{t+1}}{2}\ell(J) $$ for some $2\leq t \leq m-1$. The contribution to $v_i(x^{l-},x_{-i})$ is 
$s_t \ell(J)$ since when candidate $i$ moves infinitesimally to the left she rises one place in the rankings of these voters. The contribution to $v_i(x^{l+},x_{-i})$ is $s_{t+1}\ell(J)$, since this move causes $i$ to fall one place in these voters' rankings.

In the same way, the contribution to $v_i(x)$ from any interval $J'$ to the right of $I_l$, consisting of voters who all rank $i$ identically, is
$$
 \frac{s_t+s_{t+1}}{2}\ell(J') 
$$ 
 for some $2\leq t \leq m-1$ while the contribution to $v_i(x^{l-},x_{-i})$ is 
$s_{t+1} \ell(J')$ and the contribution to $v_i(x^{l+},x_{-i})$ is $s_{t}\ell(J')$. 

Hence, $v_i(x^{l-},x_{-i})+v_i(x^{l+},x_{-i})=2v_i(x)$ since for any subinterval $I_l$, $J$ or $J'$ the sum of the contributions to $v_i(x^{l-},x_{-i})$ and $v_i(x^{l+},x_{-i})$ is twice the contribution to $v_i(x)$ from the same subinterval.
For $x$ to be in NCNE we need both $v_i(x^{l-},x_{-i})\leq v_i(x)$ and $v_i(x^{l+},x_{-i})\leq v_i(x)$. This is only possible when $v_i(x^{l-},x_{-i})=v_i(x^{l+},x_{-i})=v_i(x)$.
\end{proof}

\begin{lemma}\label{45cand3} 
If $n_1=2$ or $n_q=2$ then a necessary condition for NCNE is $s_2=s_{m-1}$.
\end{lemma}

\begin{proof} Let $n_1=2$. By Lemma \ref{45cand2} we have $v_1(x^{1+},x_{-1})=v_1(x)$. Hence, if  1 moves to a position $t \in (x^1,x^2)$ then for NCNE we need $v_1(t,x_{-1})\leq v_1(x)=v_1(x^{1+},x_{-1})$. 
Hence, the slope of the linear function $v_1(t,x_{-1})$ is nonpositive. By Proposition~\ref{insidetheinterval2} we then have $s_2-s_{m-1}\le 0$, which can happen only if $s_2=s_{m-1}$.
\end{proof}

Propositions~\ref{insidetheinterval1} and~\ref{insidetheinterval2} allow us to conclude that unpaired candidates are actually quite rare. In particular, if $m$ is even and $s_{m/2}\ne s_{m/2+1}$, there can be none whatsoever. If $m$ is odd and $s_{(m-1)/2}\neq s_{(m+3)/2}$ then the only candidate that could possibly be unpaired is the median candidate.
Lemmas \ref{45cand2} and \ref{45cand3} tell us the only rules that allow paired candidates at the end positions are the rules of the form $s=(s_1,s_2,\ldots,s_{2},s_m)$.\par\medskip

The consequences for elections with small number of candidates are as follows: if $m=4$, since the only possible profile for NCNE is the one with two distinct positions occupied by two candidates apiece, we must have $s_2=s_3$; for $m=5$,  all possible partitions of the candidates (2-1-2 and 3-2) involve end positions occupied by exactly two candidates, hence, for NCNE we need $s_2=s_3=s_4$. For $m=6$, too, we can conclude that the partitions 2-1-1-2, 2-2-2 and 2-4 are possibilities only for rules satisfying $s_2=s_3=s_4 =s_5$ (such as plurality, for which the first two of these three partitions allow NCNE, see \cite{eatonlipsey} or \cite{denzaukatsslutsky}). This leaves only the partition 3-3 as a possible NCNE for rules which are not of this kind. We will elaborate on these results in section \ref{specialcases}. \par\medskip

\section{Rules with almost no NCNE}
\label{ruleswithnoNCNE}

In this section we will identify three quite broad classes of scoring rules for which NCNE do not exist or do not exist with a few well-defined exceptions.  

The first class is all scoring rules with convex scores. These are best-rewarding rules and hence do not allow CNE. We show that such rules do not have NCNE (in fact, they have no NE whatsoever) with the exception of some derivatives of Borda rule. 
The second class consists of rules that satisfy a certain condition on the speed with which the scores are decreasing; we call such rules weakly concave. Rules that have concave scores or symmetric scores (we will explain what this means later) belong to this class. These, in contrast, are are worst-punishing or intermediate, hence, allow CNE by Theorem \ref{CNE}.  We show that weakly concave rules with a mild additional condition do not have NCNE.  We show that if a weakly concave rule has an NCNE, then it is highly nonsymmetric. We give an example of such NCNE. We do not know, however, whether or not rules with concave scores may have NCNE. We leave this question open.
The third class consists of rules that are highly best-rewarding.

\subsection{Convex scores}

We say that the rule $s=(\row sm)$ is \textit{convex} if
\begin{equation}
\label{convexscores}
 s_1-s_2\geq s_2-s_3\geq \ldots \geq s_{m-1}-s_m. 
 \end{equation}
 We note that as soon as $s_i=s_{i+1}$ for some $i$, all the subsequent scores must also be equal for convexity to be satisfied.
 %
%
We aim to show that such rules, with one class of possible exceptions, have no NCNE; moreover they have no Nash equilibria at all. 
Firstly we show that a convex scoring rule $s$ is either best-rewarding or intermediate. In fact, we show a bit more.

\begin{prop}\label{convexscores1} 
Let $s$ be a scoring rule. Then $s$ is convex  if and only if every nonconstant $m'$-candidate subrule $s'$, where $2\leq m'\leq m$, has  $c(s',m')\geq 1/2$.
\end{prop}

\begin{proof}
Suppose $s$ satisfies \eqref{convexscores}. It suffices to show the rule itself is best-rewarding or intermediate, since any nonconstant subrule also satisfies \eqref{convexscores}. We have, for any $1\leq i\leq \lfloor m/2 \rfloor$, 
$$ s_i-s_{i+1}\geq s_{i+1}-s_{i+2}\geq \ldots \geq s_{m-i}-s_{m-i+1}. $$ 
In particular, all we need is
\begin{equation}
\label{convexscores3} 
s_i-s_{i+1} \geq s_{m-i}-s_{m-i+1}.
\end{equation}  
Suppose $m$ is even. Equation \eqref{convexscores3} implies
$$ s_1+s_m\geq s_2+s_{m-1}\geq \cdots \geq s_{m/2}+s_{m/2+1}.$$
Then
\begin{align*}\bar{s}= \frac{1}{m}\sum_{i=1}^ms_i&=\frac{(s_1+s_m)+(s_2+s_{m-1})+\cdots+(s_{m/2}+s_{m/2+1})}{m}  \\ &\leq \frac{m/2}{m}(s_1+s_m) = \frac{1}{2}(s_1+s_m).\end{align*}
Suppose $m$ is odd. Then  \eqref{convexscores3} implies
$$ s_1+s_m\geq s_2+s_{m-1}\geq \cdots \geq s_{(m-1)/2}+s_{(m-1)/2+2} \geq 2s_{(m+1)/2}.$$
Then, letting $k=(m-1)/2$, we have
\begin{align*}\bar{s}= \frac{1}{m}\sum_{i=1}^ms_i&=\frac{s_1+\cdots +s_k+s_{k+1}+\cdots+s_m }{m} \\
&=\frac{(s_1+s_m)+\cdots +(s_{k}+s_{k+2})+s_{k+1}}{m} \\
&\leq \frac{k(s_1+s_m)+s_{k+1}}{m}\\
&\leq \left(1-\frac{1}{2m}\right)(s_1+s_m)+\frac{1}{2m}(s_1+s_m) =\frac{1}{2}(s_1+s_m). \end{align*}
So in both cases $s_1+s_m\geq 2\bar{s}$, which is equivalent to $c(s,m)\geq 1/2$.

Conversely, suppose every nonconstant subrule $s'$ is best-rewarding or intermediate. Then any 3-candidate subrule $s'=(s_i,s_{i+1},s_{i+2})$ has $c(s',3)\geq 1/2$, which is equivalent to $s_i-s_{i+1}\geq s_{i+1}-s_{i+2}$, so \eqref{convexscores} is satisfied. 
\end{proof}

\begin{prop}
\label{convexscores2'} 
Let $s$ be a convex scoring rule. Then both of the following conditions:
\begin{enumerate}
\item[(a)]  all inequalities in \eqref{convexscores} are equalities,
\item[(b)] $s$ satisfies  $ s_1+s_{m}=\frac{2}{m}\sum_{i=1}^{m}s_i$,
\end{enumerate}
are equivalent to $s$ being a Borda rule.\footnote{Borda rule was defined in the introduction.}
\end{prop}

\begin{proof}
(a) Let $d$ be the common value of all the differences in \eqref{convexscores}. Then $s_i=(m-i)d+s_m$.
Subtracting $s_m$ from all scores does not change the rule. Dividing all the scores by $d$ after that does not change it either. But then we will get the canonical Borda score vector with $s_i=m-i$.

(b) The condition \eqref{convexscores} implies
\begin{equation}
\label{inequalities_of_sums}
s_1+s_m\ge s_2+s_{m-1}\ge s_3+s_{m-2}\ge \ldots
\end{equation}
from which $ s_1+s_{m}\ge \frac{2}{m}\sum_{i=1}^{m}s_i$. An equality here is possible only if we had all equalities in \eqref{inequalities_of_sums} and  this is possible only if we had equalities in \eqref{convexscores}. Now the result follows from (a).
\end{proof}

Now we can prove the main theorem of this section.

\begin{thm}
\label{convexscores2} 
Let $s$ be a scoring rule with convex scores and let $1\leq n <m$ be such that $s_n>s_{n+1}=\cdots=s_m$. Then there are no NCNE, unless the subrule $s'=(s_1,\cdots,s_n,s_{n+1})$ is Borda and $n+1\leq \lfloor m/2 \rfloor$ (i.e., more than half the scores are constant).
 \end{thm}

\begin{proof} Let $x$ be a profile. Consider candidate 1 at $x^1$. Without loss of generality, assume $n_1\leq \lfloor m/2 \rfloor$, since at least one of the two end positions has less than half the candidates. Let $I_1=[0,x^1]$ and $I_2=[x^1,(x^1+x^2)/2]$. The rest of the issue space to the right of $(x^1+x^2)/2$ can be partitioned into subintervals 
\[
J_1=\left[\frac{x^1+x^2}{2},\frac{x^1+x^3}{2}\right), \ldots, J_{j}=\left[\frac{x^1+x^{j+1}}{2},\frac{x^1+x^{j+2}}{2}\right),\ldots, J_{q-1}= \left[\frac{x^1+x^q}{2},1\right], 
\]
where voters in each of these intervals rank candidate 1 in the same way. More specifically, candidate 1 shares $k_i$-th through to $(k_i+n_1-1)$-th place in the rankings of all voters in $J_i$, for some $k_i\geq 1$ such that $k_i+n_1-1\leq m$. Then 1's score is
$$ 
v_1(x)=\left(\frac{1}{n_1}\sum_{i=1}^{n_1}s_i\right) (\ell(I_1)+\ell(I_2))+\sum_{j=1}^{q-1}\left(\frac{1}{n_1}\!\!\sum_{i=k_j}^{k_j+n_1-1}\!\!\!\!\! s_i\right)\ell(J_j). 
$$
If candidate 1 moves infinitesimally to the left, then 
$$
v_1(x^{1-},x_{-1})=s_1\ell(I_1)+s_{n_1}\ell(I_2)+\sum_{j=1}^{q-1} s_{k_j+n_1-1}\ell(J_j).
$$
Similarly, if she moves infinitesimally to the right, then 
$$
v_1(x^{1+},x_{-1})=s_1\ell(I_2)+s_{n_1}\ell(I_1)+\sum_{j=1}^{q-1} s_{k_j}\ell(J_j).
$$
Let $x$  be in NCNE. Then $v_1(x^{1-},x_{-1})\leq v_1(x)$ and $v_1(x^{1+},x_{-1})\leq v_1(x)$. This implies that $v_1(x^{1-},x_{-1})+v_1(x^{1+},x_{-1})\leq 2v_1(x)$.
That is, 
\begin{align*} 
(s_1+s_{n_1})(\ell(I_1)+\ell(I_2)&)+\sum_{j=1}^{q-1}(s_{k_j}+s_{k_j+n_1-1})\ell(J_j)\\ &\leq  \left(\frac{2}{n_1}\sum_{i=1}^{n_1}s_i\right) (\ell(I_1)+\ell(I_2))+2\sum_{j=1}^{q-1}\left(\frac{1}{n_1}\!\!\sum_{i=k_j}^{k_j+n_1-1}\!\!\!\!\!s_i\right)\ell(J_j), 
\end{align*}
which implies 
\begin{align}\label{convexscores2.1}
\left(s_1+s_{n_1}-\frac{2}{n_1}\sum_{i=1}^{n_1}s_i\right)&(\ell(I_1)+\ell(I_2))\nonumber \\
&+\sum_{j=1}^{q-1}\left(s_{k_j}+s_{k_j+n_1-1}-\frac{2}{n_1}\!\!\sum_{i=k_j}^{k_j+n_1-1}\!\!\!\!\!s_i\right)\ell(J_j)\leq 0.
\end{align}
%
%
We know that the convexity of the scores implies $s_{l}+s_{l+n_1-1}\geq \frac{2}{n_1}\sum_{i=l}^{l+n_1-1}s_i$ for all $l\geq 1$ such that $l+n_1-1\leq m$. Thus, each term on the left-hand side of \eqref{convexscores2.1} is nonnegative. If one or more of these terms is positive, then we have a contradiction and hence no NCNE exist. The only other possibility is that all these terms are equal to zero, which by Proposition~\ref{convexscores2'} implies each of the \mbox{$n_1$-candidate} subrules appearing in the expression is equal to Borda or is constant (in particular, the rule $s'=(s_1,\ldots,s_{n_1})$ must be Borda, since $s_1>s_{n_1}$ by Lemma \ref{generalprops1}). In particular we get $n_1\le n+1$. \par

If this is the case then $v_1(x^{1-},x_{-1})+v_1(x^{1+},x_{-1})=2v_1(x)$, so for $x$ to be in NCNE we must have $v_1(x^{1-},x_{-1})=v_1(x^{1+},x_{-1})=v_1(x)$. 
Then, for $x$ to be in NCNE we must have $v_1(t,x_{-1})\leq v_1(x)=v_1(x^{1+},x_{-1})$ for any $t\in (x^1,x^2)$, that is, the score cannot increase as 1 moves to the right  from $x^1$. 
By Proposition~\ref{insidetheinterval2}, the slope of the linear function $v_1(t,x_{-1})$ for $t\in (x^1,x^2)$  is $\frac{1}{2}s_{n_1}-\frac{1}{2}s_{m-n_1+1}$ and since it is nonincreasing we have $\frac{1}{2}s_{n_1}-\frac{1}{2}s_{m-n_1+1}\le  0$. On the other hand, $\frac{1}{2}s_{n_1}-\frac{1}{2}s_{m-n_1+1}\ge  0$ since $n_1<m-n_1+1$. We conclude therefore that $s_{n_1}= s_{m-n_1+1}$.
This means that the scores have stabilised on or earlier than $s_{n_1}$, whence $n_1\ge n+1$. As $n_1\le n+1$, we must now conclude that $n_1=n+1$. 

Hence, there are no NCNE unless the subrule $s'=(s_1,\ldots,s_n,s_{n+1})$ is Borda and $n+1\leq \lfloor m/2 \rfloor$.
\end{proof} 

For the special case where $s$ has convex scores, the subrule $s'=(s_1,\ldots,s_{n+1})$ is Borda and the scores from $s_{n+1}$ through $s_m$ are constant,  Theorem~\ref{convexscores2} says nothing and for good reason since here NCNE can actually exist. This will follow from  Theorem~\ref{multipositional} and Example~\ref{exceptionexamp}.\par\medskip

Rules satisfying the conditions of Theorem~\ref{convexscores2} include Borda (for which nonexistence of NCNE also follows from Theorem \ref{WPnoNCNE2}) as well as the following examples.

\begin{examp}\label{convexscoresexample} 
The following rules have convex scores and, hence, no NCNE.
\begin{enumerate}
\item[(i)] Given $s_1$, define a scoring rule by $s_i=s_1/\alpha^{i-1}$ for $2\leq i\leq m$, where $\alpha >1$. That is, multiply the previous score by the same factor $1/\alpha$ each time. 
\item[(ii)] Given $s_m$, define a rule by $s_{m-i}=s_{m-i+1}+\alpha i$ for $1\leq i\leq m-1$, where $\alpha>0$. That is, add an increasing amount each time.
\item[(iii)]   The rule $s=(1,s_2,0,\ldots,0)$, for any $0<s_2<1/2$. The significance of this is that even a slight deviation from plurality destroys the NCNE which plurality is known to possess. 
\end{enumerate}
\end{examp}

\subsection{Concave scores}
\label{WPnoNCNEsection}

We say that the rule $s=(\row sm)$ is \textit{concave} if
\begin{equation}
\label{concavescores}
 s_1-s_2\leq s_2-s_3\leq \ldots \leq s_{m-1}-s_m. 
 \end{equation}

Most our positive results are, however, applicable to a larger class of rules which we call weakly concave. We say that a scoring rule is {\it weakly concave} if it obeys the following property:
\begin{equation}\label{WPnoNCNE} s_i-s_{i+1}\leq s_{m-i}-s_{m-i+1},\end{equation}
for all $1\leq i\leq \lfloor m/2 \rfloor$. That is, the difference between consecutive scores at the top end must not be larger than the corresponding difference at the bottom end. If we always have an equality in \eqref{WPnoNCNE} we say that the rule is \textit{symmetric}.

\begin{prop}\label{WPnoNCNE1} A weakly concave rule  is either worst-punishing or  intermediate. That is, $c(s,m)\leq 1/2$.\end{prop}

\begin{proof} Note that \eqref{WPnoNCNE} is condition \eqref{convexscores3} with the inequalities reversed. Hence, reversing all the inequalities in Proposition \ref{convexscores1}, we obtain $c(s,m)\leq 1/2$.
\end{proof}

Before we can prove the main result of this section, we will need one more lemma.

\begin{lemma}\label{WPnoNCNElemma} 
If $s$ is a weakly concave rule, then
\begin{equation}
\label{WPnoNCNElemma.1}
 s_j+s_{m-j+1}\geq \frac{1}{j}\left(\sum_{i=1}^js_i+\!\!\!\!\!\sum_{i=m-j+1}^m\!\!\!\!\!s_i\right)
 \end{equation}
for all $1\leq j \leq \lfloor m/2 \rfloor$. Moreover, if  $s$  satisfies 
\begin{equation}
\label{WPnoNCNElemma.2} 
s_k+s_{m-k+1}\geq \frac{1}{k}\left(\sum_{i=1}^ks_i+\!\!\!\!\!\sum_{i=m-k+1}^m\!\!\!\!\!s_i\right)
\end{equation} 
for some $k>\lfloor m/2 \rfloor$, then inequality \eqref{WPnoNCNElemma.1} holds for all $1\leq j\leq k$. 
\end{lemma}

\begin{proof} Let $1\leq j \leq \lfloor m/2\rfloor$. Equation \eqref{WPnoNCNE} implies that 
$$ s_1+s_m\leq s_2+s_{m-1}\leq \cdots \leq s_{j}+s_{m-j+1},$$
whence $$ \sum_{i=1}^{j}s_i+\!\!\!\sum_{i=m-j+1}^m\!\!\!\!s_i=\sum_{i=1}^{j}(s_i+s_{m-i+1}) \leq  j(s_{j}+s_{m-j+1}),$$
which, on dividing by $j$, gives \eqref{WPnoNCNElemma.1}. This proves the first part of the lemma.

Now suppose that \eqref{WPnoNCNElemma.2} holds for some $k>\lfloor m/2\rfloor$. The statement will be proved by induction if we can prove that  \eqref{WPnoNCNElemma.1} holds for $j=k-1$. If $j=\lfloor m/2\rfloor$ the statement follows from the first part of the lemma. So assume $k>\lfloor m/2\rfloor+1$.  We have 
\begin{equation*}
\left(\sum_{i=1}^ks_i+\!\!\!\!\!\!\sum_{i=m-k+1}^m\!\!\!\!s_i\right)-\left(\sum_{i=1}^{k-1}s_i+\!\!\!\!\!\!\sum_{i=m-k+2}^m\!\!\!\!s_i\right)=s_k+s_{m-k+1} 
\geq  \frac{1}{k}\left(\sum_{i=1}^ks_i+\!\!\!\!\!\!\sum_{i=m-k+1}^m\!\!\!\!s_i\right).
\end{equation*}
This rearranges to give
$$ \frac{1}{k}\left(\sum_{i=1}^ks_i+\!\!\!\sum_{i=m-k+1}^m\!\!\!\!s_i\right)\geq \frac{1}{k-1}\left(\sum_{i=1}^{k-1}s_i+\!\!\!\sum_{i=m-k+2}^m\!\!\!\!s_i\right).$$
Since $k>\lfloor m/2 \rfloor +1$, we have $m-k+1\leq \lfloor m/2\rfloor$ and hence by \eqref{WPnoNCNE} we conclude  $s_{m-k+1}+s_k\leq s_{m-k+2}+s_{k-1}$. Putting things together, we get
$$ s_{m-k+2}+s_{k-1}\geq \frac{1}{k-1}\left(\sum_{i=1}^{k-1}s_i+\!\!\!\sum_{i=m-k+2}^m\!\!\!\!s_i\right).$$
Thus, equation \eqref{WPnoNCNElemma.2} holds for $j=k-1$, which proves the induction step.
\end{proof}

Next, we show a weakly concave rule  has no NCNE in which each of the end positions is occupied by less than half the candidates. If \eqref{WPnoNCNElemma.2} also holds for $k=m-3$, which is condition \eqref{WPnoNCNE2.1}  below, 
then we can rule out NCNE altogether. 

\begin{thm}\label{WPnoNCNE2} 
Any weakly concave scoring rule $s$  has no NCNE in which $\max(n_1,n_q) \leq \lfloor m/2 \rfloor$. If, in addition, $s$ satisfies 
\begin{equation}\label{WPnoNCNE2.1} 
s_4+s_{m-3}\geq \frac{1}{m-3}\left(\sum_{i=1}^{m-3}s_i+\sum_{i=4}^ms_i\right),
\end{equation}  then no NCNE exist.
 \end{thm}

\begin{proof} 
Suppose $\max(n_1,n_q) \leq \lfloor m/2 \rfloor$. Consider candidate 1 at position $x^1$, which is occupied by $n_1$ candidates. Consider intervals $I_1=[0,x^1]$ and $I_2=[(x^1+x^q)/2,1]$. If 1 makes an infinitesimal move to the right of $x^1$, then in the rankings of voters in $I_1$ she falls behind the other $n_1-1$ candidates originally at $x^1$. On the other hand, 1 rises ahead of these $n_1-1$ candidates in the rankings of all other voters.
Then the score candidate 1 loses by making this move, $s_{lost}$, is
$$ s_{lost}=\left(\frac{1}{n_1}\sum_{i=1}^{n_1}s_i-s_{n_1}\right)\ell(I_1).$$
On the other hand, 1's gain from this move, $s_{gained}$, is \textit{at least} the gain from $I_2$:
\begin{equation}\label{WPnoNCNE2.2} s_{gained}\geq \left(s_{m-n_1+1}-\frac{1}{n_1}\!\sum_{i=m-n_1+1}^m\!\!\!\!\!s_i\right)\ell(I_2) \geq \left(\frac{1}{n_1}\sum_{i=1}^{n_1}s_i-s_{n_1}\right)\ell(I_2),\end{equation}
where we have used \eqref{WPnoNCNElemma.1}. For this profile to be an NCNE, we need this move not be beneficial for candidate 1. That is, we need $s_{lost}\geq s_{gained}$, or
$$ \left(\frac{1}{n_1}\sum_{i=1}^{n_1}s_i-s_{n_1}\right)\ell(I_1) \geq  \left(\frac{1}{n_1}\sum_{i=1}^{n_1}s_i-s_{n_1}\right)\ell(I_2). $$
Since by Lemma \ref{generalprops1} we know that $s_1>s_{n_1}$, and hence the common multiple on both sides of the inequality is nonzero, this implies
\begin{equation}\label{leftside}
 \ell(I_1)\geq \ell(I_2) \mbox{\quad or \quad}  x^1\geq 1-\frac{x^1+x^q}{2}.
 \end{equation}
Since $n_q$ is also assumed to be not greater than $\lfloor m/2 \rfloor$, similar considerations with respect to candidate $q$ give that $\ell([x^q,1])\geq \ell([0,(x^1+x^q)/2]) $. That is,
\begin{equation}\label{rightside} 
1-x^q\geq \frac{x^1+x^q}{2}.
\end{equation}
Together, \eqref{leftside} and \eqref{rightside} imply that $x^1\geq x^q$, which is impossible for an NCNE. Hence, there are no NCNE in which $\max(n_1,n_q)\leq \lfloor m/2\rfloor$.

Now, if \eqref{WPnoNCNE2.1} holds as well, then by \eqref{WPnoNCNElemma.2} of Lemma \ref{WPnoNCNElemma} we conclude that 
\begin{equation}\label{WPnoNCNE2.3} s_{m-n_1+1}+s_{n_1}\geq \frac{1}{n_1}\left(\sum_{i=1}^{n_1}s_i+\!\!\!\sum_{i=m-n_1+1}^m\!\!\!\!\!s_i\right)\end{equation} for all $2\leq n_1 \leq m-3$, which are all the possible values for $n_1$, with the exception of the case  $n_1=m-2$. However, if $n_1=m-2$ then $n_q=2$, in which case Lemma \ref{45cand3} implies that $s_2=s_{m-1}$. But for any rule of this form, \eqref{WPnoNCNE2.3} is satisfied for $n_1=m-2$, also.

Hence, inequalities \eqref{leftside} and \eqref{rightside} hold even if $n_1> \lfloor m/2\rfloor$, and similarly for $n_q> \lfloor m/2\rfloor$. So there can be no NCNE at all. 
\end{proof}

The following rules will appear important later.

\begin{cor}
\label{ale2b}
Any scoring rule $s=(a,b,\ldots,b,0)$ for $a\le 2b$ does not have an NCNE.
\end{cor}

\begin{proof}
These rules are weakly concave and condition~\eqref{WPnoNCNE2.1} is satisfied.
\end{proof}

Let us now have a look at some more examples. 

\begin{examp}\label{WPnoNCNEexample} 
When $m=8$, it turns out that the additional condition \eqref{WPnoNCNE2.1} in Theorem \ref{WPnoNCNE2} follows from  \eqref{WPnoNCNE}. Indeed,  it is equivalent to
$$ 
s_4+s_5\geq \frac{1}{3}(s_1+s_2+s_3+s_6+s_7+s_8).
$$
But by \eqref{WPnoNCNE} we have $s_1+s_8\leq s_2+s_7 \leq s_3+s_6\leq s_4+s_5$, hence \eqref{WPnoNCNE2.1} always holds if  \eqref{WPnoNCNE} holds. 

By a similar argument,  condition \eqref{WPnoNCNE2.1} is again redundant when $m\leq 8$. If $m=9$, though, there are rules that satisfy \eqref{WPnoNCNE} but not \eqref{WPnoNCNE2.1}. Consider the rule $s=(7,6,6,6,6,2,1,1,0)$. We have $s_4+s_6=8<8\frac{1}{6}=\frac{1}{6}(\sum_{i=1}^6s_i+\sum_{i=4}^9s_i)$. If an NCNE exists for this rule, it must have more than four candidates at one of the end positions. However, computationally we find this rule has no NCNE.
\end{examp}

 An interesting special case is rules with symmetric scores.

\begin{cor}\label{WPnoNCNEsymmetric} Any symmetric scoring rule $s$  has no NCNE.\end{cor}

\begin{proof} Clearly \eqref{WPnoNCNE} is satisfied. Condition \eqref{WPnoNCNE2.1} is also satisfied since for any valid value of $n_1$ we have
\begin{align*}\frac{1}{n_1}\!\sum_{i=m-n_1+1}^{m}\!\!\!\!\!s_i+\frac{1}{n_1}\sum_{i=1}^{n_1}s_i=\frac{1}{n_1}\sum_{i=1}^{n_1}(s_i+s_{m-i+1})&=\frac{1}{n_1}\sum_{i=1}^{n_1}(s_{n_1}+s_{m-n_1+1})\\
&=s_{n_1}+s_{m-n_1+1}.\qedhere \end{align*}
\end{proof}

Some examples of rules with symmetric scores are single-positive and single-negative voting, given by $s=(2,1,\ldots,1,0)$, and the rule $s=(4,3,2,\ldots,2,1,0)$. \par

Finally, we will  give an example showing that condition \eqref{WPnoNCNE2.1} in Theorem \ref{WPnoNCNE2} is necessary. 

\begin{examp}
\label{8-4NCNE}
 For $m=12$  the scoring rule $s=(4,4,4,3,3,3,2,1,1,0,0,0)$ satisfies weak concavity, yet does allow NCNE. In particular, the profile 
$$
((x^1,n_1),(x^2,n_2))=\left(\left(\frac{13}{28},8\right),\left(\frac{41}{84},4\right)\right)
$$ 
with eight candidates at position $x^1=\frac{13}{28}$ and four at position $x^2=\frac{41}{84}$ is an NCNE.
\end{examp}

\begin{proof}
It is clear this rule satisfies weak concavity. We check that it does not satisfy the additional condition:
%
$$s_4+s_{9}=4< \frac{38}{9}= \frac{1}{9}\left(\sum_{i=1}^{9}s_i+\sum_{i=4}^{12}s_i\right).$$
So we cannot rule out highly asymmetric equilibria.

Now, we check the given profile is an NCNE.
Consider first candidate 1, located at $x^1$.
Her score in the given profile is 
$$
v_1(x)=\frac{1}{8}\left[24 \left(\frac{x^1+x^2}{2}\right)+10\left(1-\frac{x^1+x^2}{2}\right)\right]=\frac{25}{12}.
$$
We must consider all possible moves by candidate 1.
Moving slightly to the left is not profitable since
$$v_1(x^{1-},x_{-1})=4x^1+\left(\frac{x^2-x^1}{2}\right)=\frac{157}{84}<\frac{25}{12}=v_1(x).$$
All moves into the middle interval are dominated by the move to $x^{1+}$, but this is not profitable since
$$v_1(x^{1+},x_{-1})=4\left(\frac{x^2-x^1}{2}\right)+x^1+3\left(1-\frac{x^1+x^2}{2}\right)=\frac{25}{12}= v_1(x).$$
A move to $x^{2+}$ is not profitable as
$$
v_1(x^{2+},x_{-1})=4(1-x^2)+3\left(\frac{x^2-x^1}{2}\right)=\frac{25}{12}= v_1(x).
$$
Finally, a move to $x^2$ is not profitable either, since
$$
v_1(x^2,x_{-1})=\frac{1}{5}\left[18\left(1-\frac{x^1+x^2}{2}\right)+2\left(\frac{x^1+x^2}{2}\right)\right]=\frac{218}{105}<v_1(x).
$$
So candidates at $x^1$ cannot improve their scores.

Now, check for candidate 9, at position $x^2$. We have
$$v_9(x)=\frac{1}{4}\left[15\left(1-\frac{x^1+x^2}{2}\right)+\left(\frac{x^1+x^2}{2}\right)\right]=\frac{25}{12}.$$
The  move to $x^{2+}$ gives 
$$v_9(x^{2+},x_{-9})=4(1-x^2)+3\left(\frac{x^2-x^1}{2}\right)=\frac{25}{12}= v_9(x).$$
All moves into the middle interval are dominated by the move to $x^{1+}$, which is not profitable since
$$v_9(x^{1+},x_{-9})=4\left(\frac{x^2-x^1}{2}\right)+x^1+3\left(1-\frac{x^1+x^2}{2}\right)=\frac{25}{12}=v_9(x).$$
The move to $x^{1-}$ gives
$$v_9(x^{1-},x_{-9})=4x^1+\left(\frac{x^2-x^1}{2}\right)=\frac{157}{84}<v_9(x).$$
Finally, the move to $x^1$ gives
$$
v_9(x^1,x_{-9})=\frac{1}{9}\left[25\left(\frac{x^1+x^2}{2}\right)+13\left(1-\frac{x^1+x^2}{2}\right)\right]=\frac{131}{63}<v_9(x).
$$
There are no other moves to consider, so the given profile is an NCNE.
\end{proof}

Concavity is stronger than weak concavity, however it does not  imply \eqref{WPnoNCNE2.1} either. A counterexample is $s=(4,4,4,4,4,4,4,4,4,4,4,4,3,2,1,0)$ when $m=16$. It has concave scores but does not satisfy \eqref{WPnoNCNE2.1}. Of course, it has no NCNE by Corollary~\ref{generalprops2}. \par\medskip

This leads us to a question which we leave open in this paper: does there exist a concave rule that has an NCNE?\par\medskip


A few words about a dichotomy between CNE and NCNE---when $m=4$, a given rule cannot have both NCNE and CNE.

\begin{cor}\label{dichotomy} 
If $m=4$ and $s=(s_1,s_2,s_3,s_4)$ is a scoring rule such that $c(s,m)\leq 1/2$, then there are no NCNE. 
\end{cor}

\begin{proof}
Let $m=4$ and suppose $s=(s_1,s_2,s_3,s_4)$ is a worst-punishing rule. That is, $$c(s,4)=\frac{s_1-\frac{1}{4}(s_1+s_2+s_3+s_4)}{s_1-s_4}\leq \frac{1}{2}, $$
which is equivalent to $s_1-s_2\leq s_3-s_4 .$
That is, condition \eqref{WPnoNCNE} is equivalent to $c(s,m)\leq 1/2$. Condition \eqref{WPnoNCNE2.1}, as we commented in Example \ref{WPnoNCNEexample}, is redundant for $m\leq 8$ (indeed, if $m=4$, the only possible type of an NCNE is (2,2), so we cannot have any asymmetry), so NCNE cannot exist. 
\end{proof}

For five candidates and more, however, $c(s,m)\leq 1/2$ does not imply condition \eqref{WPnoNCNE}---there exist rules that are worst-punishing and do not satisfy \eqref{WPnoNCNE}. Nevertheless, as we will see later, in Theorem \ref{5cand2},  the dichotomy  for five candidates still holds and no rule can have both CNE and NCNE. However
 this dichotomy need not hold in general---indeed when there are six candidates it fails: this will be illustrated in Example~\ref{bipositionalexample}.

\subsection{Highly best-rewarding rules}
\label{highlyBRsection}

We will now show that for many rules, if the value of $c(s,m)$ is high enough, then no NCNE are possible. 

\begin{thm}
\label{highlyBR} 
If $m$ is even and $s$ satisfies both 
\begin{equation}
\label{highlybestr1}
c(s,m)>1-\frac{1}{m-2} \mbox{\quad and \quad} s_{m/2}\ne s_{m/2+1},
\end{equation}
or, if $m$ is odd and $s$ satisfies both 
\begin{equation}
\label{highlybestr2}
c(s,m)>1-\frac{1}{m-1} \mbox{\quad and \quad} s_{(m-1)/2}\ne s_{(m+3)/2},
\end{equation}
then $s$ allows no NCNE.
 \end{thm}

\begin{proof} Suppose first that $m$ is even and $s$ satisfies \eqref{highlybestr1}. Rearranging, we get
$$
 \frac{1}{2(1-c(s,m))}>\frac{m}{2}-1.
 $$
By Corollary~\ref{lowerboundq}, then, we have $q\geq \lceil \frac{1}{2(1-c(s,m))}\rceil \geq m/2$. Also, since $s_{m/2}\ne s_{m/2+1}$, we know by Proposition~\ref{insidetheinterval2} that in an NCNE there are no unpaired candidates, and hence there must be at least two candidates at each occupied position. But there are at least $m/2$ positions, hence there must be two candidates at each one, including the end positions. But also $s_2\geq s_{m/2}\ne s_{m/2+1}\geq s_{m-1}$, which violates Lemma \ref{45cand3}. So we cannot have NCNE.

Now suppose $m$ is odd. Similarly to the above we get $$\frac{1}{2(1-c(s,m))}>\frac{m-1}{2},$$ hence $q\geq (m+1)/2$. There must be at least one position with only one candidate and in fact there can be only one such position since, by Proposition~\ref{insidetheinterval2}, the only candidate that can be unpaired is the median candidate. This leaves at least two candidates at each of the $q-1\geq (m-1)/2$ other occupied positions. Again, the only possibility is that there are two candidates at each of these remaining positions, but this contradicts Lemma \ref{45cand3}.
\end{proof}

Let us look at a concrete examples of rules that satisfy Theorem \ref{highlyBR}. Clearly all such rules are best-rewarding (recall that we assume $m\geq 4$), hence these rules allow no Nash equilibria of any kind.

\begin{examp} When $m=6$, the rule is $ s=(s_1,s_2,s_3,s_4,s_5,0)$ and the conditions of Theorem~\ref{highlyBR}  are equivalent to $s_3\ne s_4$ and $s_1>2(s_2+s_3+s_4+s_5)$. The rules  $s=(5,1,1,0,0,0)$ and $s=(7,2,1,0,0,0)$ are of this kind.
%
\end{examp}

Note that the rules in the example above may have convex scores, so we see that this class of rules overlaps with the class of rules appearing in Theorem~\ref{convexscores2}. Clearly neither class contains the other, though. 
If the second condition in \eqref{highlybestr1} or \eqref{highlybestr2} is not satisfied, then there may be NCNE, even if the rule has a very high value of $c(s,m)$. Plurality, for example, has the maximal value of $c(s,m)$ but still possesses NCNE.

\section{Rules allowing NCNE}
\label{ruleswithNCNE}

In this section we will turn our attention to rules for which, unlike the rules considered up to now, NCNE do exist.
We first look at a class of rules $s=(a,b,\ldots,b,0)$ for which we find experimentally as many NCNE as for plurality and with similar properties. Second, we look at the class of best-rewarding rules for which we can find NCNE in which candidates cluster at positions spread across the issue space. Then, we characterise symmetric bipositional  NCNE when $m$ is even.

\subsection{Rules of the type $(a,b,\ldots,b,0)$}
\label{rulesabbb0}

We have seen in Corollary~\ref{ale2b} that for $a\le 2b$ this rule does not have NCNE. The situation changes radically if we assume that $a>2b$.  This is not surprising due to the following

\begin{prop}
\label{propabbb}
The $m$-candidate rule $s=(a,b,\ldots,b,0)$  is best rewarding if and only if $a>2b$ and worst punishing or intermediate when $a\le 2b$.
\end{prop}

\begin{proof}
We have 
\[
c(s,m)=\frac{a-\frac{a+(m-2)b}{m}}{a}=\frac{m-1}{m}-\frac{m-2}{m}\cdot \frac{b}{a}.
\]
Then $c(s,m)\le \frac{1}{2}$ is equivalent to $(m-2)\le 2(m-2)\frac{b}{a}$, from which the statement follows.
\end{proof}

 Computational experiments show that  when $a>2b$ the rule $s=(a,b,\ldots,b,0)$ has multiple NCNE similar to those discovered for plurality by Eaton and Lipsey \cite{eatonlipsey} and Denzau et al. \cite{denzaukatsslutsky}.  For example, the rule $s=(3,1,1,1,1,1,1,0)$ has NCNE with $q=4,5,6$ and $(n_1,\ldots, n_q)$ being one of the following:
\[
(2,2,2,2),\ (2,2,1,1,2),\ (2,1,2,1,2),\ (2,1,1,2,2),\ (2,1,1,1,1,2).
\]
We notice that at any given position there are no more than two candidates. 

\begin{thm}
\label{nomorethantwo}
If $s=(\row sm)=(a,b,\ldots,b,0)$, then in any NCNE  at any given position there are no more than two candidates. 
\end{thm}

\begin{proof}
By Corollary~\ref{ale2b} we have to consider only the case when $a>2b$. This case will be considered in the following two lemmas.
\end{proof}
 
 \begin{lemma} If $s=(\row sm)=(a,b,\ldots,b,0)$, where $a>2b$, then in NCNE $n_i\leq 2$ for all $2\leq i \leq q-1$.\end{lemma}

\begin{proof} If $n_i>2$, where $2\leq i\leq q-1$, then candidate $k$, located at $x^i$ is not ranked last by any voter. Moreover, she is not ranked last by any voter even on deviating to $x^{i+}$ or $x^{i-}$. So the only change in her score on making these moves is from voters in the immediate subintervals $I_1=[(x^{i-1}+x^i)/2,x^{i}]$ and $I_2=[x^{i},(x^{i}+x^{i+1})/2]$, where  voters change candidate $k$ from first equal to first, and from first equal to $n_i$th, respectively. 
In NCNE we must have $$v_k(x^{i-},x_{-k})-v_k(x)=s_1 \ell(I_1)+s_2\ell( I_2)-\left(\frac{1}{n_i}\sum_{j=1}^{n_i}s_j\right)(\ell(I_1)+\ell(I_2))\leq 0.$$
Also $$v_k(x^{i+},x_{-k})-v_k(x)=s_1\ell(I_2)+s_2\ell(I_1)-\left(\frac{1}{n_i}\sum_{j=1}^{n_i}s_j\right)(\ell(I_1)+\ell(I_2))\leq 0.$$
Adding together these two inequalities we get
$s_1+s_{n_i} \leq \frac{2}{n_i}\sum_{j=1}^{n_i}s_j,$
but the subrule $s'=(\row s{n_i})=(a,b,\ldots,b)$ of length $n_i$ is best-rewarding whenever $n_i>2$. That is, $s_1+s_{n_i}> \frac{2}{n_i}\sum_{j=1}^{n_i}s_j$.
\end{proof}

\begin{lemma}
\label{n1=2} 
If $s=(\row sm)=(a,b,\ldots,b,0)$, where $a>2b$, then in NCNE  $n_1=n_q=2$.
\end{lemma}

\begin{proof} Let us introduce the following notation:
\begin{itemize}
\item $\alpha=x^1$ -- the proportion of voters to the left of candidate 1.
\item $\beta=(x^2-x^1)/2$ -- the proportion of voters in half the interval between candidates 1 and 2.
\item $\gamma= 1-(x^1+x^q)/2$ -- the proportion of voters for whom 1 is ranked last equal.
\end{itemize}Note that
$$
v_1(x)=\frac{s_1}{n_1}(\alpha+\beta)+s_2-\frac{s_2}{n_1}\left(\alpha+\beta+\gamma)\right)=\frac{(s_1-s_2)}{n_1}(\alpha+\beta)+s_2\left(1-\frac{\gamma}{n_1}\right).
$$
Consider if 1 moves to $x^{1-}$. Then  
$$
v_1(x^{1-},x_{-1})=s_1\alpha+s_2(1-\alpha-\gamma)=(s_1-s_2)\alpha+s_2(1-\gamma).
$$
If 1 moves to $x^{1+}$ then 
$$
v_1(x^{1+},x_{-1})=s_1\beta+s_2(1-\beta)=(s_1-s_2)\beta+s_2.
$$
For NCNE we require that these moves not be beneficial to candidate 1.
That is, $v_1(x^{1-},x_{-1})\leq v_1(x)$ which implies
\[
 (s_1-s_2)\alpha+s_2(1-\gamma) \leq \frac{(s_1-s_2)}{n_1}(\alpha+\beta)+s_2\left(1-\frac{\gamma}{n_1}\right)
 \]
 or
 \begin{equation}
\label{n1=2.1}
(s_1-s_2)\left(1-\frac{1}{n_1}\right)\alpha \leq \frac{(s_1-s_2)}{n_1}\beta+s_2\left(1-\frac{1}{n_1}\right)\gamma.
\end{equation}
Similarly, for the other move we have $v_1(x^{1+},x_{-1})\leq v_1(x)$ which gives us
\[
 (s_1-s_2)\left(1-\frac{1}{n_1}\right)\beta\leq \frac{(s_1-s_2)}{n_1}\alpha-s_2\frac{\gamma}{n_1}
 \]
 or
 \begin{equation}
 \label{n1=2.2}
 (s_1-s_2)\left[\left(1-\frac{1}{n_1}\right)\beta-\frac{\alpha}{n_1}\right]\leq -s_2\frac{\gamma}{n_1} \leq0.
 \end{equation}
Then \eqref{n1=2.2} implies that $\beta(n_1-1)\leq \alpha$ since we know that $s_1>s_2$. 
In addition, rearranging \eqref{n1=2.2} gives
\[ 
s_2\gamma \leq (s_1-s_2)(\alpha+\beta-n_1\beta) 
\]
and multiplying through by the positive number $1-1/n_1$
\begin{equation}
\label{n1=2.3} 
s_2\left(1-\frac{1}{n_1}\right)\gamma \leq \left(1-\frac{1}{n_1}\right)(s_1-s_2)(\alpha+\beta-n_1\beta).\end{equation}
Substituting \eqref{n1=2.3} into \eqref{n1=2.1} and dividing through by $s_1-s_2>0$ we get 
\begin{align} \left(1-\frac{1}{n_1}\right)\alpha&\leq \frac{\beta}{n_1}+\left(1-\frac{1}{n_1}\right)(\alpha+\beta-n_1\beta)\nonumber\\
\label{n1=2.4} 0&\leq \beta(2-n_1) .\end{align}  
Since $\beta>0$, equation \eqref{n1=2.4} requires that $n_1\leq 2$. So $n_1=2$. 
A similar argument gives $n_q=2$.
\end{proof}

\subsection{Multipositional clustered NCNE}

We introduce some additional notation. Let
\begin{enumerate}
\item[(i)] $I_1=[0,(x^1+x^2)/2]$,
\item[(ii)] $I_i=[(x^{i-1}+x^{i})/2,(x^{i}+x^{i+1})/2]$ for $2\leq i \leq q-1$,
\item[(iii)]  $I_q=[(x^{q-1}+x^{q})/2,1]$,
\end{enumerate}
 be the ``full-electorates" corresponding each occupied position. For each $i\in [q]$ let $I_i^L=\{y\in I_i: y\leq x^i\}$ and $I_i^R=\{y\in I_i: y\geq x^i\}$ be the left and right ``half-electorates" whose union is the full-electorate $I_i$, that is $I_i=I_i^L\cup I_i^R$.  We note that  $\ell(I_i^R)=\ell(I_{i+1}^L)$ for $i\in [q-1]$. \par\medskip

The following theorem provides a method of constructing rules for which multipositional NCNE exist.

\begin{thm}
\label{multipositional} 
Let $m=qr$ be a composite number with $q\geq 2$.  Consider an $m$-candidate scoring rule 
$s=(s_1,\ldots,s_{r-1},0,\underbrace{0,\ldots,0}_{r}, \cdots ,\underbrace{0,\ldots,0}_{r})$. 
Then the profile given by $x=((x^1,r),\ldots,(x^q,r))$ is in NCNE if and only if  the following two conditions hold:
\begin{enumerate}
\item[(a)] $\max_{i \in [q]}\max\{\ell(I_i^L),\ell(I_i^R)\} \leq (1-c(s',r))\min_{i \in [q]}\{\ell(I_i)\}$,
\item[(b)] $\max_{i \in [q]}\{\ell(I_i)\}\leq \left(1+\frac{1}{r}\right)\min_{i \in [q]}\{\ell(I_i)\}$,
\end{enumerate}
where $s'=(s_1,\ldots,s_{r-1},0)$.
\end{thm}

The idea is that for this kind of scoring rule, each occupied position is ``isolated" from the rest of the issue space, since a candidate at this position receives nothing from voters who rank her $r$th or worse. So the candidates have to compete ``locally''. Note that condition (a) can only be satisfied if $c(s',r)\leq 1/2$, since it implies 
\begin{align*}
\max_{i \in [q]}\max\{\ell(I_i^L),\ell(I_i^R)\} &\leq (1-c(s',r))\min_{i \in [q]}\{\ell(I_i)\} \\
&\leq 2(1-c(s',r))\max_{i \in [q]}\max\{\ell(I_i^L),\ell(I_i^R)\},
\end{align*}
from which  $c(s',r)\leq 1/2$ follows. That is, though the scoring rule is best-rewarding, the subrule $s'$, for $x$ to be in NCNE, must be worst-punishing or intermediate. Hence, comparing this with Theorem \ref{CNE}, we see that locally each occupied position behaves with respect to the rule $s'$ in a similar way to a CNE on the whole issue space.

\begin{proof}
Consider candidate $i$ at position $x^k$. Since all of $i$'s score is garnered from the immediate full-electorate $I_k$, $i$'s score is
$$v_i(x)=\left(\frac{1}{r}\sum_{j=1}^rs_j\right)\ell(I_k).$$
Suppose that $i$ moves to some position $t$ between two occupied positions or between an occupied position and the boundary of the issue space. In the latter case, $i$ is now ranked first by, at best, all voters in the intervals $I_1^L$ or $I_q^R$. In the former case, when $x^{l}<t<x^{l+1}$ for some $l$, candidate $i$ is ranked first by voters in the interval $[(x^{l}+t)/2,(t+x^{l+1})/2]$, which is equal in length to $\ell(I_{l}^R)=\ell(I_{l+1}^L)$. From the rest of the issue space, $i$ is ranked at best $r$th, so receives nothing.
In each case, $i$'s score is now 
$$
v_i(t,x_{-i})=s_1\ell(J),
$$ 
for the half-electorate $J$ that $i$ moves into. For NCNE we need this move not be beneficial, that is, $v_i(t,x_{-i})\leq v_i(x)$. Thus, for NCNE we must have 
$$s_1\ell(J)\leq \left(\frac{1}{r}\sum_{j=1}^rs_j\right)\ell(I_k),$$
which occurs if and only if
$$ 
\ell(J)\leq 
(1-c(s',r))\ell(I_k).
$$
This must hold when $i$ moves into any of the half-electorates, and for any candidate at any initial position. This yields the necessity of condition (a).

Also, there is a possibility that $i$ moves to some position $x^l$ that is already occupied. In this case her score becomes
$$
v_i(x^l,x_{-i})=\left(\frac{1}{r+1}\sum_{j=1}^{r+1}s_j\right)\ell(I_l)=\left(\frac{1}{r+1}\sum_{j=1}^{r}s_j\right)\ell(I_l),
$$
which again must not exceed $v_i(x)$. Hence, we must have
$$  
\left(\frac{1}{r+1}\sum_{j=1}^{r}s_j\right)\ell(I_l) \leq \left(\frac{1}{r}\sum_{j=1}^rs_j\right)\ell(I_k) 
$$
or
$$ 
\ell(I_l)\leq \left(1+\frac{1}{r}\right)\ell(I_k).
$$
Again, this must hold for any pair of full-electorates, which implies
that condition (b) is necessary.

There are no other possible moves, so (a) and (b) are sufficient for NCNE.
\end{proof}

We note that the degree to which the positions can be nonsymmetric depends on how small $c(s',r)$ is. If $c(s',r)=1/2$, for example, then by condition (a) of Theorem~\ref{multipositional} we must have that all the electorates are the same size and the occupied positions are at the halfway point of each one. If the profile is symmetric, with the candidates positioned so as to divide the issue space into equally sized full-electorates, Theorem~\ref{multipositional} simplifies.

\begin{cor}\label{multipositionalsymmetric} Let there be $m=qr$ candidates, $q\geq 2$, and consider the scoring rule $s=(s_1,\ldots,s_{r-1},0,\underbrace{0,\ldots,0}_{r}, \cdots ,\underbrace{0,\ldots,0}_{r})$. Then the profile given by $x=((x^1,r),\ldots,(x^q,r))$ such that $\ell(I_i)=1/q$ for all $i\in [q]$, is in NCNE if and only if $c(s',r)\leq 1/2$, where $s'=(s_1,\ldots,s_{r-1},0)$.\end{cor}

\begin{proof} Since each $I_i$ is the same length, condition (b) of Theorem \ref{multipositional} is satisfied. Condition (a) reduces to $1/2q \leq (1-c(s',r))/q$, whence the requirement that $c(s',r)\leq 1/2$.\end{proof}

Now let us look at some examples.

\begin{examp}  Consider 
$r$-candidate $k$-approval rule $s'=(\underbrace{1,\ldots,1,0,\ldots,0}_{r})$ with $r\geq k+1$. The condition $c(s',r)\leq 1/2$ holds if and only if $r\leq 2k$, suppose this is true. By appending zeros to the end of $s'$, we can extend $s'$ to $k$-approval with $m=qr$ candidates for any $q\geq 2$. Then Theorem \ref{multipositional} implies there exist NCNE in which $r$ candidates position themselves at each of the $q$ distinct locations. 

As a special case, consider 1-approval, which is just plurality: $s=(1,0,\ldots,0)$. For any even $m$, if we set $r=2$ then we obtain  $s'=(1,0)$ with $c(s',2)=1/2$. So the profile where two candidates locate at each position so as to divide the space into equally sized intervals is an NCNE, and it is the only one in which there are two candidates at each position. We cannot have $r>2$, as then we would have $c(s',r)>1/2$. So plurality has no equilibria in which more than two candidates locate at each position, in agreement with the well-known results of Eaton and Lipsey \cite{eatonlipsey} and Denzau et al. \cite{denzaukatsslutsky}. 
\end{examp}

\begin{examp}\label{exceptionexamp} 
Let $s'=(r-1,r-2,\ldots,2,1,0)$, that is, $s'$ is Borda. Let $s$, of length $m=qr$, $q\geq 2$, be the rule resulting from appending $(q-1)r$ zeros to~$s'$. Then $c(s',r)=1/2$, so there exists an NCNE in which $r$ candidates position themselves at the $q$ halfway points of $q$ equally sized full-electorates that partition the issue space. 
\end{examp}

Recall that Theorem~\ref{convexscores2} stated that a rule with convex scores has no NCNE, unless the nonconstant part of the scoring rule is exactly Borda and is shorter than the constant part.  The rule $s$ in Example~\ref{exceptionexamp} is precisely such a rule. Hence, the exception in Theorem~\ref{convexscores2} does indeed need to be made.

 \begin{examp} For a given scoring rule, NCNE with different partitions of the candidates can exist simultaneously. Consider 3-approval, $s=(1,1,1,0,\ldots,0)$, with $m=20$. Then if $r=4$ we have $c((1,1,1,0),4)=1/4<1/2$, so there are NCNE with five distinct positions occupied by four candidates apiece. At the same time, if $r=5$ we have $c((1,1,1,0,0),5)=2/5<1/2$, so there are also NCNE with four distinct positions occupied by five candidates apiece. 
\end{examp}

\subsection{Symmetric bipositional NCNE}

A CNE is the simplest kind of Nash equilibrium that may exist. We now turn our attention to what would be the next simplest kind -- an NCNE in which there are only two occupied positions. Here we  restrict ourselves to the case where $m$ is even and the equilibrium positions are symmetric. We saw in Example~\ref{8-4NCNE} that bipositional equilibria are not necessary symmetric. At the end of this section we will give an example of a nonsymmetric bipositional equilibria for $m=7$. Later it will become clear that this is the smallest value of $m$ for which nonsymmetric bipositional equilibria exist.

\begin{thm}\label{bipositional} Suppose $m$ is even. Then the profile $x=((x^1,m/2),(x^2,m/2))$, with $0<x^1<1/2$ and $x^2=1-x^1$, is in NCNE if and only if both
\begin{equation}
\label{bipositional1} 
\frac{s_{m/2}+s_{m/2+1}}{2}<  \bar{s} 
\end{equation} 
and
\begin{equation}
\label{bipositional2} 
\frac{s_1+s_{m/2}-2 \bar{s}}{2(s_1-s_{m/2+1})}\leq x^1 \leq \frac{2 \bar{s}-s_m-s_{m/2}}{2(s_1-s_{m/2})}.
\end{equation}
If in addition  $c(s,m)\leq 1/2$, then  the profile $x$ is in NCNE whenever
\begin{equation}
\label{bipositionalcorollary1} 
\frac{s_1+s_{m/2}-2\bar{s}}{2(s_1-s_{m/2+1})}\leq x^1 <\frac{1}{2} .
\end{equation} 
Moreover, \eqref{bipositionalcorollary1} can always be satisfied.
\end{thm}

\begin{proof} By the symmetry of the positions, $(x^1+x^2)/2=1/2$ and $(x^2-x^1)/2=1/2-x^1$. At $x$, all candidates receive $1/m$-th of the points, so $v_i(x)=\bar{s}$ for all $i=1,\ldots,m$. Note that it is necessary that $s_1>s_{m/2}$, since otherwise there would need to be more than $m/2$ candidates at each position by Lemma \ref{generalprops1}. 

By symmetry, for NCNE it is enough to require candidate 1 not be able to deviate profitably, and there are only three moves to consider: a move to $x^1-\epsilon$, which is always better than a move to $x^2+\epsilon$, since 1 is ranked one place higher for half of voter in the middle interval; a move to $x^2-\epsilon$, which is the best move out of any into the middle interval since the slope of $v_1(t,x_{-1})$ in that interval is nonnegative by Proposition~\ref{insidetheinterval2}; and, finally, a move to $x^2$.

For the first one we have
$$ 
v_1(x^{1-},x_{-1})=s_1x^1+s_{m/2}\left(\frac{1}{2}-x^1\right)+\frac{1}{2}s_m.
$$
For NCNE, it must be that $v_1(x^{1-},x_{-1})\leq v_1(x)$, which yields the requirement
\begin{equation}
\label{bipositional1.1} 
x^1\leq  \frac{2\bar{s}-s_m-s_{m/2}}{2(s_1-s_{m/2})}.
 \end{equation}
For the second move we have
$$
v_1(x^{2-},x_{-1})=s_1\left(\frac{1}{2}-x^1\right)+\frac{1}{2}s_{m/2}+s_{m/2+1}x^1.
$$
The fact that $v_1(x^{2-},x_{-1})\leq v_1(x)$ yields 
\begin{equation}
\label{bipositional1.2} 
x^1\geq \frac{s_1+s_{m/2}-2\bar{s}}{2(s_1-s_{m/2+1})}.
\end{equation}
Finally,
$$
v_1(x^2,x_{-1})=\frac{1}{m/2+1}\left(\sum_{i=1}^{m/2+1}\!\!\! s_i+\sum_{i=m/2}^m\!\!\!s_i\right)\frac{1}{2}=\frac{1}{m+2}\left(m\bar{s}+s_{m/2}+s_{m/2+1}\right).
$$
Since in an NCNE $v_1(x^2,x_{-1})\leq v_1(x)$, it must be that $s_{m/2}+s_{m/2+1}\leq 2\bar{s}.$ 

There are no other moves to consider, so if the position $x^1$ is valid, that is, satisfies \eqref{bipositional1.1}, \eqref{bipositional1.2} and is in the range $0<x^1<1/2$, then we have an NCNE.
The condition $x^1<1/2$ combined with \eqref{bipositional1.2} implies the strict inequality in \eqref{bipositional1}. The condition $x^1>0$ means that we need the right-hand side of \eqref{bipositional1.1} to be strictly greater than zero, which is always true.

Finally, note that the requirement that $s_1>s_{m/2}$ is implied by \eqref{bipositional1}, since if $s_{m/2}=s_1$ then we have
$$
 s_{m/2}+s_{m/2+1}=s_1+s_{m/2+1}<2\bar{s}=s_1+\frac{2}{m}\!\sum_{i=m/2+1}^m\!\!\!\!\!s_i\leq s_1+s_{m/2+1},
$$
a contradiction.

To prove the second statement, suppose a scoring rule satisfies both $c(s,m)\leq 1/2$ and \eqref{bipositional1}. Since $c(s,m)\leq 1/2$ is equivalent to $s_1+s_m\leq 2\bar{s}$, the right-hand side of \eqref{bipositional2} always satisfies
\begin{equation}
\label{bipositionalcorollary.1}
\frac{1}{2}\leq \frac{2\bar{s}-s_m-s_{m/2}}{2(s_1-s_{m/2})}.
\end{equation}
Similarly, the left-hand side satisfies 
\begin{equation}
\label{bipositionalcorollary.2}
\frac{s_1+s_{m/2}-2\bar{s}}{2(s_1-s_{m/2+1})} <\frac{1}{2}.
\end{equation}
Putting together \eqref{bipositionalcorollary.1} and \eqref{bipositionalcorollary.2}, we see that it will always be possible to find valid values of $x^1$ in the desired range.
 \end{proof}
 
 \begin{examp} Again consider $k$-approval with $k<m/2$. Clearly, \eqref{bipositional1} is satisfied. Then we have symmetric bipositional NCNE whenever $1/2-k/m\leq x^1\leq k/m$, which is valid whenever $k\geq m/4$. Thus, as $k$ decreases and the rule becomes more best-rewarding, the more extreme positions are possible, until we reach the point where a bipositional equilibrium is no longer viable. NCNE with more than two positions then become possible, as can be seen by Theorem \ref{multipositional} and as one would expect by Corollary \ref{lowerboundq}.\end{examp}
 
  Theorem \ref{bipositional} allows us to conclude that bipositional NCNE may exist for both best-rewarding and worst-punishing rules, as we will see in the examples below. 
  
\begin{examp}
\label{bipositionalexample} 
Let $m=6$. Consider the following rules:
\begin{enumerate}
\item[(i)] $s=(2,2,1,1,1,0)$. We have $s_3+s_4=2<7/3=\frac{2}{6}\sum_{i=1}^6s_i$, so \eqref{bipositional1} is satisfied. Equation \eqref{bipositional2} reduces to $1/3\leq x^1 \leq 2/3$, so the profile $x=((x^1,3),(1-x^1,3))$ is in NCNE for any $1/3\leq x^1<1/2$. Note that $c(s,6)=5/12<1/2$, so this rule is worst-punishing.
\item[(ii)] $s=(10,10,4,3,3,0)$. We have $s_3+s_4=7<10=\frac{2}{6}\sum_{i=1}^6s_i$, so \eqref{bipositional1} is satisfied. By equation \eqref{bipositional2}, NCNE occurs whenever $2/7\leq x^1<1/2$. This rule is intermediate since $c(s,6)=1/2$.
\item[(iii)] $s=(4,3,1,1,0,0)$. We have $s_3+s_4=2<3=\frac{2}{6}\sum_{i=1}^6s_i$, so \eqref{bipositional1} is satisfied and we have NCNE when $x^1=1/3$. So this rule allows only one symmetric bipositional NCNE. Here we have $c(s,6)=5/8>1/2$, so $s$ is best-rewarding.
\end{enumerate}
The first two rules also allow CNE, so we see that CNE and NCNE can coexist for the same rule. The third rule, on the other hand, has no CNE.
\end{examp}

Finally, we give an example of a bipositional equilibrium in which the number of candidates is different at the two positions, though the positions themselves turn out to be symmetrically located.

\begin{examp}
\label{nonsymbipositionalexample} 
Let $m=7$. Consider the rule $s=(10, 10, 4, 3, 3, 1, 0)$. Then the profile $((x^1,4), (x^2,3))$  with $x^1=1/3$ and $x^2=2/3$, is an NCNE. We omit the details.
\end{examp}

\section{The four-, five- and six-candidate cases}
\label{specialcases}

In the special cases of $m=4$ and $m=5$  we can provide a complete characterisation of the rules allowing NCNE.  For $m=6$ we can identify all types of possible equilibria.

\begin{thm}
\label{4cand} 
Given $m=4$ and scoring rule $s=(s_1,s_2,s_3,s_4)$, NCNE exist if and only if both the following conditions are satisfied:
\begin{enumerate}
\item[(a)] $c(s,4)>1/2$,
\item[(b)] $s_1>s_2=s_3$.
\end{enumerate}
Moreover, the NCNE is unique and symmetric, with equilibrium profile $x=((x^1,2),(x^2,2))$, where
\begin{equation}
\label{candidatepositions} 
x^1=\frac{1}{4}\left(\frac{s_1-s_4}{s_1-s_2}\right) \mbox{\quad and \quad} x^2=1-x^1.
\end{equation} 
\end{thm}

\begin{proof} By Lemma~\ref{generalprops1}, an NCNE with $m=4$ must have exactly two distinct positions, $x^1<x^2$, with $n_1=n_2=2$. Hence, by Lemma \ref{45cand3}, it is  necessary that $s_2=s_3$. By Lemma \ref{generalprops1}, we also need $s_1>s_2$. Hence, (b) is necessary.

By Lemma \ref{noextremepositions}, we have $0<x^1$ and $x^2< 1$. By Lemma \ref{45cand2},  in NCNE we have
$v_1(x^{2+},x_{-1})\leq v_1(x^{1-},x_{-1})=v_1(x)$. That is,
\begin{align*} 
v_1(x^{2+},x_{-1})&=s_1(1-x^2)+s_2\left(\frac{x^2-x^1}{2}\right)+s_4\left(\frac{x^1+x^2}{2}\right)\\ 
            &\leq s_1x^1+s_2\left(\frac{x^2-x^1}{2}\right)+s_4\left(1-\frac{x^1+x^2}{2}\right) =v_1(x^{1-},x_{-1}), 
\end{align*}
which implies $s_1(1-x^1-x^2)\leq s_4(1-x^1-x^2)$. This is only possible if $x^1+x^2\geq 1$. Considering the symmetric moves by candidate 4 gives $(1-x^1)+(1-x^2)\geq 1$ or $x^1+x^2\leq 1$. Hence, $x^2=1-x^1$. 

Then, since $v_1(x^{1-},x_{-1})=v_1(x)$, we have
$$ s_1x^1+s_2\left(\frac{x^2-x^1}{2}\right)+\frac{1}{2}s_4 = \frac{1}{4}s_1+\frac{1}{2}s_2+\frac{1}{4}s_4,$$
from which, after substituting for $x^2$, equation \eqref{candidatepositions} follows. For this to be a valid position, we need $x^1<1/2$, from which it follows that $2s_2<s_1+s_4$. This is equivalent to $c(s,4)>1/2$, so (a) is necessary.

For sufficiency, notice that a rule satisfying conditions (a) and (b) also satisfies the conditions of Theorem \ref{bipositional}, from which we conclude that the profile given by \eqref{candidatepositions} is actually in NCNE. \end{proof}

For the five-candidate case, there are two ways to partition the candidates that might result in an NCNE. As it turns out, one of them is not possible.

\begin{lemma}\label{5cand1} For $m=5$, there are no NCNE of the form $x=((x^1,2),(x^2,3))$.\end{lemma}

\begin{proof} First note that, without loss of generality, we can assume $s_5=0$, since it is easy to see that subtracting $s_5$ from each score does not change the rule.  
By Lemma \ref{45cand3} we have $s_2=s_3=s_4$ and by Lemma \ref{generalprops1} we have $s_1>s_2$. Hence our rule is one of those studied in Subsection~\ref{rulesabbb0}. But then by Lemma~\ref{n1=2} we cannot have three candidates at position $x^2$.
\end{proof}


\begin{thm}\label{5cand2} Given $m=5$ and scoring rule $s=(s_1,s_2,s_3,s_4,s_5)$, NCNE exist if and only if both the following conditions are satisfied:
\begin{enumerate}
\item[(a)] $c(s,5)>1/2$,
\item[(b)] $s_1>s_2=s_3=s_4$.
\end{enumerate}
Moreover, the NCNE is unique and symmetric, with equilibrium profile $x=((x^1,2),(1/2,1),(x^3,2))$, where
\begin{equation}\label{5candpositions} x^1=\frac{1}{6}\left(\frac{s_1+s_2}{s_1-s_2}\right) \mbox{\quad and \quad} x^3=1-x^1.\end{equation}
\end{thm}

\begin{proof} As in Lemma \ref{5cand1}, we lose no generality in assuming $s_5=0$. 
By Lemma \ref{5cand1} and \ref{generalprops1}, the profile must be of the form $x=((x^1,2),(x^2,1),(x^3,2))$. Also, $s_2=s_3=s_4$ by Lemma \ref{45cand3} and $s_1>s_2$ by Lemma \ref{generalprops1}, so condition (b) is necessary. By Lemma \ref{noextremepositions}, the end points of the issue space are not occupied.

As in the proof of Theorem~\ref{4cand}, considering moves by candidate 1 to $x^{1-}$ and $x^{3+}$, together with Lemma \ref{45cand2}, gives $x^1 \geq 1-x^3$. Similar considerations for candidate 5 give $x^1\leq 1-x^3$, hence $x^1=1-x^3$.

Let  $t\in(x^1,x^2)$ and $t'\in (x^2,x^3)$  (all positions in these intervals yield the same score by Proposition~\ref{insidetheinterval2}). Again by Lemma~\ref{45cand2}, we need
\begin{align*} 
v_1(t',x_{-1})&=s_1\left(\frac{x^3-x^2}{2}\right)+s_2\left(1-\frac{x^3-x^2}{2}\right)\\
&\leq v_1(x)=v_1(t,x_{-1})=s_1\left(\frac{x^2-x^1}{2}\right)+s_2\left(1-\frac{x^2-x^1}{2}\right),
\end{align*}
which implies $x^3-x^2\leq x^2-x^1$ and hence $\frac{1}{2}(x^1+x^3)\leq x^2$. The same considerations with respect to candidate~$5$ give that $x^2\leq \frac{1}{2}(x^1+x^3)$. So we have equality and, consequently, $x^2=1/2$.

We know that $v_1(x^{1-},x_{-1})=v_1(x)=v_1(x^{1+},x_{-1})$. This yields
$$ 
s_1x^1\!+s_2\!\left(\frac{x^3-x^1}{2}\right)\!=s_1\!\left(\frac{x^2-x^1}{2}\right)\!+s_2\!\left(1-\frac{x^2-x^1}{2}\right),
$$
from which, after substituting $x^2=1/2$ and $x^3=1-x^1$, equation \eqref{5candpositions} follows. For this to be a valid position, we need $x^1<x^2=1/2$. This gives $s_1>2s_2$, which is equivalent to $c(s,5)>1/2$, so condition (a) is necessary.

Now sufficiency. Suppose (a) and (b) are satisfied. Note that neither candidate~$1$ nor candidate~$5$ can move into $[0,x^1)$, $(x^1,1/2)$, $(1/2,x^3)$ or $(x^3,1]$ beneficially. Again, this follows by symmetry and by Lemma \ref{45cand2}. We check the remaining possibilities.
As
$$ 
v_1(x^2,x_{-1})=\frac{1}{6}s_1+\frac{2}{3}s_2 = v_1(x),
$$
the move by candidate~1 to $x^2$ is not beneficial. Also 
$$
v_1(x^3,x_{-1})=\frac{1}{6}(s_1-s_2)x^1+\frac{1}{12}s_1+\frac{3}{4}s_2=\frac{1}{9}s_1+\frac{7}{9}s_2\leq \frac{1}{6}s_1+\frac{2}{3}s_2=v_1(x),
$$
so there is no reason for candidate~1 to move to $x^3$. By symmetry, then, no moves by candidate~5 are beneficial.

Finally, consider moves by candidate~3. Any move inside the  interval $(x^1,x^3)$ does not change her score. A move to a position infinitesimally to the left of $x^1$ gives the requirement
$$
v_3(x^{1-},x_{-3})=\frac{1}{6}s_1+\frac{2}{3}s_2\leq \frac{1}{3}(s_1+s_2)=v_3(x),
$$
which is satisfied. Finally, if candidate~3 moves to $x^1$ her score is
$$v_3(x^1,x_{-3})=\frac{1}{6}s_1+\frac{2}{3}s_2<v_3(x).$$
By symmetry, then, no moves are beneficial for this candidate.
There are no more moves to consider, hence this is an NCNE. \end{proof}

Thus, in both the four- and five-candidate cases, we see that NCNE exist only for a subset of best-rewarding rules -- those for which all scores except first and last are equally valuable. In both cases, the amount of dispersion observed in the candidates' positions depends on the difference between $s_1$ and $s_2$ and is maximal when $s_2=0$, that is, when the rule is plurality.

As $s_2$ grows towards $s_1/2$, the positions of candidates become less extreme, converging at the median voter position when $s_2=s_1/2$. As $s_2$ increases beyond this point, by Theorem \ref{CNE} we know that infinitely  many CNE are possible in an interval that becomes increasingly wide. Hence, there is a
bifurcation point that divides CNE from NCNE when $c(s,4)=1/2$ or $c(s,5)=1/2$. As we move away from this point, more extreme positions are possible -- on one side they take the form of CNE, and on the other side they are NCNE.\footnote{It is curious to note that this equilibrium behaviour shows certain similarities to some of the equilibria numerically calculated by De Palma et al. \cite{depalmahongthisse} using a probabilistic model. When the level of uncertainty is low (which corresponds roughly to when $c(s,m)$ is large), they observe NCNE where the candidates are configured as in our NCNE. As the level of uncertainty increases from zero (the value of $c(s,m)$ decreases), they also observe the candidates' positions becoming less extreme. Beyond a certain point, only convergent equilibria are observed. We note, however, that in addition to these, they also observe other kinds of equilibria that do not arise in our model.}

Since for $m>5$ the equilibria are no longer unique even for plurality \cite{eatonlipsey}, it makes sense to describe only their types.

\begin{thm}\label{6cand2} Given $m=6$ and scoring rule $s=(s_1,s_2,s_3,s_4,s_5,s_6)$.  Then there are four possible types of equilibria split in two groups:
\[
\{(2,2,2),\ (2,1,1,2)\}\ \text{and}\ \{ (3,3), (6)\}.
\]
The equilibria of the first group occur for rules $s$ that satisfy
\begin{enumerate}
\item[(a)] $c(s,6)>1/2$,
\item[(b)] $s_1>s_2=s_3=s_4=s_5$.
\end{enumerate}
The equilibria within each group can coexist. No equilibrium of the first group can coexist with an equilibrium of the second group.
\end{thm}

\begin{proof} We have to show that equilibria of types
\begin{equation}
\label{falseeq}
(2,4),\ (2,1,3),\ (3,1,2),\ (4,2)
\end{equation}
do not exist. What they all have in common is that they have two candidates at one of the extreme positions. Suppose such an equilibrium exists.
As before we note that, without loss of generality, we can assume $s_6=0$.  
By Lemma \ref{45cand3} we have $s_2=s_3=s_4=s_5$ and by Lemma \ref{generalprops1} we have $s_1>s_2$. Hence, our rule is one of those studied in Subsection~\ref{rulesabbb0}. But then by Theorem~\ref{nomorethantwo} we cannot have three or more candidates at any given position. This rules out all equilibria \eqref{falseeq} and shows that equilibria of the first group are incompatible with equilibria of the second. 

Example~\ref{bipositionalexample}(i) demonstrates that equilibria of the third type exist and may coexist with CNE as the rule in this case can be worst punishing. Equilibria of the second type are shown to exist for plurality in \cite{eatonlipsey}.
\end{proof}

\section{Computational Results}
\label{comres}

We have developed an algorithm to determine whether an NCNE exists for
a given $m$-candidate scoring rule. The algorithm works by generating a
list of all possible clusterings of the candidates. For each clustering it
produces a linear program (LP). The variables of the LP are the political positions of
the clusters. The LP has the basic constraints to ensure that all of the
political positions are in order and between $0$ and $1$. Note that a
standard LP does not have strict inequalities, but we can maximise the
minimum distance between the positions so that we will find a solution
where clusters have distinct political positions, if such a solution
exists.

The non-trivial constraints are those used to ensure that the
candidates cannot improve their score by switching to a different
position. We expand the reals with political positions $x^-$ and $x^+$
that occur immediately before and after each variable $x$. For each
pair of variables $x$, $y$ we have three constraints requiring that
the score a candidate would gain from moving from $y$ to $x^+$, $x$,
or $x^-$ must be at most zero. We know that if a candidate can improve their
score, they can improve it by moving to one of these positions. While  $x^-$ (or $x^+$) is not a real number, we know that if a
candidate can improve their score by moving to $x^-$ they can also
improve their score by moving to $x-\epsilon$ for a sufficiently small
real number $\epsilon$.

This algorithm is not polynomial. Although LP solving can be
polynomial, the number of clusters considered is not polynomial. In spite of this, for the small sample problems considered in this paper the performance of the algorithm is close to instantaneous. The algorithm only demonstrates whether
a particular scoring rule gives an NCNE. We intend to adapt the
algorithm to use a Quadratic Constraint solver in place of the LP to
allow us to find whether a class of rules has an NCNE. The class of concave rules will be the prime target.

\section{Related Literature}
\label{relatedliterature}

The distinguishing features of this paper are the use of scoring rules and the focus on multicandidate nonconvergent equilibria. We now briefly describe the related literature with respect to these aspects.

In light of the probabilistic interpretation of the scoring rule mentioned in the introduction, we feel the need to emphasise the differences between our approach and the probabilistic voting literature, which is a huge subject in its own right. For an introduction to the field, see Coughlin \cite{coughlin} or, for a survey, Duggan \cite{duggan}. To the best of our knowledge, these works usually involve some distance dependent function to give the probabilities as, for example, in multinomial logit choice models. These functions lead to an expected vote share that depends on distance and does not have discontinuities when candidates' positions coincide. A scoring rule, on the other hand, is somewhat simpler: it behaves like a step function that only depends on ordinal information. Whether the second ranked candidate is just beyond the first ranked candidate or on the other side of the issue space has no bearing on the probability of the voter voting for the more distant candidate. Moreover, the discontinuities associated with the deterministic model persist in our model. Other models include the stochastic model of Anderson et al. \cite{andersonkatsthisse} and Enelow et al. \cite{linenelowdorrusen}. The former again involves a function depending on distance, and the latter involves a finite number of voters in a multidimensional space. De Palma et al. \cite{depalmahongthisse} look at a model incorporating uncertainty and numerically calculate equilibria for up to six candidates, comparing them with the deterministic model. Interestingly, some of the observed equilibria for four and five candidates show similarities to those we find in our model (see the footnote following Theorem \ref{5cand2}).

Scoring rules and similar voting systems have appeared in spatial models before.  However, apart from Cox \cite{cox1}, all the work has been done in somewhat different contexts.  Myerson \cite{myerson1} looks at the incentives inherent in different scoring rules and the political implications for such matters as corruption, barriers to entry and strategic voting. His model consists of a simpler issue space, with candidates deciding between two policy positions -- ``yes" or ``no". In \cite{myerson3}, Myerson compares various scoring rules with respect to the campaign promises they encourage candidates to make and, in \cite{myerson2}, he investigates scoring rules from the voter's perspective in Poisson voting games. Laslier and Maniquet \cite{lasliermaniquet} look at multicandidate elections under approval voting when the voters are strategic. Myerson and Weber \cite{myerson4} introduce the concept of a ``voting equilibrium", where voters take into account not only their personal preferences but also whether contenders are serious, and compare plurality and approval voting in a three-candidate positioning game similar to ours.  We focus only on candidate strategies.

As for the multicandidate aspect,  this paper is most closely related to Denzau et al. \cite{denzaukatsslutsky} and, before them, Eaton and Lipsey \cite{eatonlipsey}. The latter consider plurality rule, and the former extend these results to ``generalised rank functions''. That is, the candidates' objectives depend to some degree both on market share and rank (say, the number of candidates with a larger market share). For a review emphasising multicandidate competition, see Shepsle \cite{shepsle}.

Many other more realistic refinements of the Hotelling model have been constructed, incorporating uncertainty, incomplete information, incumbency and underdog effects, and so on. However, the price of added realism is that they are more complicated, and considering more than two candidates is often intractable. For a survey focusing on variations of the two candidate case, see Duggan \cite{duggan} or Osborne \cite{osborne2}.

As for the economic interpretation, our results are related only to those models in which price competition and transport costs do not come into play, such as in, again, Eaton and Lipsey \cite{eatonlipsey} and Denzau et al. \cite{denzaukatsslutsky}. The most basic model incorporating price competition is a game of two interdependent stages, the location selecting stage and the price setting stage. The details can be found in many standard economics textbooks, such as Vega-Redondo \cite[pp.~171--176]{vegaredondo}. 

\section{Conclusion}
\label{conclusion}

In this paper, we have investigated how the particular scoring rule in use influences the candidates' position-taking behaviour. We have looked at the equilibrium properties of a number of different classes of scoring rules. 
We were able to identify several broad classes of scoring rules disallowing NCNE completely. For other large classes, we found that NCNE can exist and we calculated a number of them. 

As Cox \cite{cox1} and Myerson \cite{myerson1} found previously, the parameter $c(s,m)$ plays a prominent role in determining what kind of equilibria are possible---as $c(s,m)$ increases, the amount of dispersion tends to increase also---though usually the value of $c(s,m)$ is not the only factor of importance. The manner in which the scores in the score vector are decreasing is pivotal---we saw, for example, that all rules with entirely convex scores and many with entirely concave scores fail to possess NCNE.  The conditions under which NCNE do exist can at times be quite stringent, as is made clear by the four- and five-candidate cases. 

In his investigation of CNE, Cox \cite{cox1} found that the value $c(s,m)=1/2$ appears as a cut-off point between existence and nonexistence of CNE: they exist if and only if $c(s,m)\leq 1/2$, i.e., when the rule is worst-punishing. When there are four or five candidates, a similar phenomenon occurs with respect to NCNE: only when $c(s,m)>1/2$ may NCNE exist, though, as mentioned above, it is not guaranteed. Thus, the two kinds of equilibrium are mutually exclusive. When there are six or more candidates, however, this transition from CNE to NCNE is no longer clear-cut. There exist worst-punishing and intermediate rules that allow both types of equilibria, as seen in Theorem~\ref{bipositional}.
 
A number of questions remain open. Though we have investigated a wide variety of scoring rules, these are by no means all of them. The equilibrium behaviour of many rules---most notably concave ones---remains unknown. On the other hand, many of the rules most frequently appearing in the literature---plurality, Borda, $k$-approval and so on---fall nicely into the cases we have considered.  
Another point of interest: most NCNE discovered in this paper by theoretical considerations have been ones in which the same number of candidates locate at each position.  The mechanics behind less regular or asymmetric equilibria obtained as a result of computational experiments remains unclear. It is interesting that a class of weakly concave rules can have only highly asymmetric equilibria. Hence a concave rule may have only asymmetric equilibria or none at all. At this point it is unknown whether any concave rules actually permit such asymmetric NCNE.

Of course, there are a number of simplifications in our framework in comparison to Cox's one. The main one is the assumption that the voters are uniformly distributed along the issue space, though this was partially justified in the first footnote of Section \ref{themodel}. Cox's \cite{cox1} characterisation of CNE holds for an arbitrary nonatomic distribution of voter ideal points. The existence of NCNE, however, is much more vulnerable to changes in the distribution. In a similar framework, Osborne \cite{osborne1} finds that, for plurality rule, the uniform distribution is a special case---``almost all" other distributions exclude the possibility of NCNE. It would be interesting to
see whether our results can be adapted to other distributions.
Another limiting assumption is the unidimensionality of the issue space. Cox \cite{cox1} provides a version of Theorem~\ref{CNE} for a multidimensional space. Whether our results on NCNE can be extended to this situation is also not known.

\end{document}